\documentclass[11pt]{article}
\usepackage{times}
\usepackage{helvet}
\usepackage{courier}
\usepackage{amsmath}
\usepackage{amsfonts}
\usepackage{amssymb}
\usepackage{amsthm}
\usepackage[left=2cm,right=2cm,top=2cm,bottom=2cm]{geometry}
\usepackage{graphicx}
\usepackage{tikz}
\usepackage{algorithm}
\usepackage{algorithmic}
\usepackage{hyperref}
\usepackage{url}
\usepackage{mathtools}
\usepackage{natbib}
\numberwithin{equation}{section}

\newtheorem{thm}{Theorem}

\newtheorem{prop}{Proposition}

\newcommand{\vect}[1]{\boldsymbol{#1}}
\renewcommand{\bar}{\overline}
\newcommand{\eps}{\epsilon}
\newcommand{\pa}{\partial}

\renewcommand{\eps}{\varepsilon}
\renewcommand{\epsilon}{\varepsilon}
\renewcommand{\Sigma}{\varSigma}

\newcommand{\E}{\textrm E}

\begin{document}
%
\title{Wormhole Hamiltonian Monte Carlo}
\author{Shiwei Lan\footnote{Department of Statistics, University of California,
Irvine, USA.}\\ \texttt{slan@uci.edu} \and Jeffrey Streets\footnote{epartment of
Mathematics, University of California, Irvine, USA}\\ \texttt{jstreets@uci.edu}
\and Babak Shahbaba\footnote{Department of Statistics and Department of Computer
Science, University of California, Irvine, USA.}\\ \texttt{babaks@uci.edu}}
\maketitle
\begin{abstract}
\begin{quote}
In machine learning and statistics, probabilistic inference involving multimodal distributions is quite difficult. This is especially true in high dimensional problems, where most existing algorithms cannot easily move from one mode to another. To address this issue, we propose a novel Bayesian inference approach based on Markov Chain Monte Carlo. Our method can effectively sample from multimodal distributions, especially when the dimension is high and the modes are isolated. To this end, it exploits and modifies the Riemannian geometric properties of the target distribution to create \emph{wormholes} connecting modes in order to facilitate moving between them. Further, our proposed method uses the regeneration technique in order to adapt the algorithm by identifying new modes and updating the network of wormholes without affecting the stationary distribution. To find new modes, as opposed to rediscovering those previously identified, we employ a novel mode searching algorithm that explores a \emph{residual energy} function obtained by subtracting an approximate Gaussian mixture density (based on previously discovered modes) from the target density function. 
\end{quote}
\end{abstract}

\section{Introduction}
In Bayesian inference, it is well known that standard Markov Chain Monte Carlo (MCMC) methods tend to fail when the target distribution is
multimodal \citep{neal93, neal96a, celeux00, neal01, rudoy06, sminchisescu11, craiu09}. These methods typically fail to move from one mode to another
since such moves require passing through low probability regions. This is especially true for high dimensional problems with isolated modes.
Therefore, despite recent advances in computational Bayesian methods, designing effective MCMC samplers for multimodal distribution has remained a
major challenge. In the statistics and machine learning literature, many methods have been proposed address this issue \citep[see for example,][]{neal96a, neal01, warnes01, laskey03, hinton04, braak06, rudoy06, sminchisescu11, ahn13} However, these methods tend to suffer from the curse
of dimensionality \citep{hinton04, ahn13}.

In this paper, we propose a new algorithm, which exploits and modifies the Riemannian geometric properties of the target distribution to create
wormholes connecting modes in order to facilitate moving between them. Our method can be regarded as an extension of Hamiltonian Monte Carlo (HMC).
Compared to random walk Metropolis, standard HMC explores the target distribution more efficiently by exploiting its geometric properties. However, it
too tends to fail when the target distribution is multimodal since the modes are separated by high energy barriers (low probability regions)
\citep{sminchisescu11}.

In what follows, we provide an brief overview of HMC. Then, we introduce our method assuming that the locations of the modes are known (either exactly or approximately), possibly through some optimization techniques (e.g., \citep{kirkpatrick83, sminchisescu02}). Next, we relax this assumption by incorporating a mode searching algorithm in our method in order to identify new modes and to update the network of wormholes.

\section{Preliminaries}

Hamiltonian Monte Carlo (HMC) \citep{duane87,neal10} is a Metropolis algorithm with proposals guided by Hamiltonian dynamics. HMC improves upon random
walk Metropolis by proposing states that are distant from the current state, but nevertheless have a high probability of acceptance. These distant
proposals are found by numerically simulating Hamiltonian dynamics, whose
state space consists of its \emph{position}, denoted by the vector $\boldsymbol\theta$, and its \emph{momentum}, denoted by a vector $\boldsymbol p$.
Our objective is to sample from the distribution of $\boldsymbol\theta$ with the probability density function (up to some constant)
$\pi(\boldsymbol\theta)$. We usually assume that the auxiliary momentum variable $\boldsymbol p$ has a multivariate normal distribution (the same
dimension as $\boldsymbol\theta$) with mean zero. The covariance of $\boldsymbol p$ is usually referred to as the \emph{mass matrix}, $\boldsymbol M$,
which in standard HMC is usually set to the identity matrix, ${\boldsymbol I}$, for convenience.

Based on $\boldsymbol \theta$ and $\boldsymbol p$, we define the \emph{potential energy}, $U(\boldsymbol\theta)$, and the \emph{kinetic energy},
$K(\boldsymbol p)$. We set $U(\boldsymbol \theta)$ to minus the log probability density of $\boldsymbol \theta$ (plus any constant). For the auxiliary
momentum variable $\boldsymbol p$, we set $K(\boldsymbol p)$ to be minus the log probability density of $\boldsymbol p$ (plus any constant). The
\emph{Hamiltonian} function is then defined as follows:
\begin{eqnarray*}\label{hamiltonian}
H(\boldsymbol \theta,\boldsymbol p) & = & U(\boldsymbol\theta) + K(\boldsymbol p)
\end{eqnarray*}
The partial derivatives of $H(\boldsymbol \theta, \boldsymbol p)$ determine how $\boldsymbol \theta$ and $\boldsymbol p$ change over time, according
to \emph{Hamilton's equations},
\begin{eqnarray}\begin{array}{lcrcr }
\displaystyle 
\dot {\boldsymbol\theta} & = & \nabla_{\boldsymbol p} H({\boldsymbol\theta}, {\boldsymbol p}) & = & {\boldsymbol M}^{-1}{\boldsymbol p} \\ [12pt]
\displaystyle 
\dot {\boldsymbol p} & = & -\nabla_{\boldsymbol\theta} H({\boldsymbol\theta}, {\boldsymbol p})& =  & -\nabla_{\boldsymbol\theta} U(\boldsymbol\theta) 
\end{array}\label{rmhd}\end{eqnarray}
Note that ${\boldsymbol M}^{-1}{\boldsymbol p}$ can be interpreted as velocity. 

In practice, solving Hamiltonian's equations exactly is difficult, so we need to approximate these equations by discretizing time, using some small
step size $e$. For this purpose, the \emph{leapfrog} method is commonly used. We can use some number, $L$, of these leapfrog steps, with some
stepsize, $e$, to propose a new state in the Metropolis algorithm. This proposal is either accepted or rejected based on the Metropolis acceptance probability. 

While HMC explores the target distribution more efficiently than random walk Metropolis, it does not fully exploits its geometric properties. Recently, \cite{girolami11} proposed a new method, called Riemannian Manifold HMC (RMHMC), that improvs the efficiency of standard HMC by automatically adapting to the local structure. To this end, they follow \cite{amari00} and propose
Hamiltonian Monte Carlo methods defined on the Riemannian manifold endowed with metric ${\boldsymbol G}_{0}({\boldsymbol\theta})$, which is set to the Fisher information matrix. More specifically, they define Hamiltonian dynamics in terms of a position-specific mass matrix, $\boldsymbol M$, set to ${\boldsymbol G}_{0}({\boldsymbol\theta})$. The standard HMC method is a special case of RMHMC with ${\boldsymbol G}_{0}({\boldsymbol\theta}) =
{\boldsymbol I}$. Here, we use the notation ${\boldsymbol G}_{0}$ to generally refer to a Riemannian metric, which is not necessarily the Fisher information. In the following section, we introduce a natural modification of ${\boldsymbol G}_{0}$ such that the associated Hamiltonian dynamical system has a much greater chance of moving between isolated modes.

\section{Wormhole Hamiltonian Monte Carlo}

Consider a manifold $\mathcal M$ endowed with a generic metric ${\boldsymbol G}_{0}({\boldsymbol\theta})$. Given a differentiable curve
$\boldsymbol{\theta}(t) :
[0,T] \to \mathcal M$ one can define the arclength along this curve as
  \begin{equation}\label{arcL}
  \ell({\boldsymbol\theta}) := \int_0^T \sqrt{
  \dot{\boldsymbol\theta}(t)^{\textsf T} {\boldsymbol G}_{0}({\boldsymbol\theta}(t))
  \dot{\boldsymbol\theta}(t)} dt
  \end{equation}
Under very general geometric assumptions, which are nearly always satisfied in statistical models, given any two points ${\boldsymbol\theta}_1,
{\boldsymbol\theta}_2 \in \mathcal M$ there exists a curve ${\boldsymbol\theta}(t):[0, T]\to \mathcal M$ satisfying the boundary conditions
${\boldsymbol\theta}(0) = {\boldsymbol\theta}_1, {\boldsymbol\theta}(T) = {\boldsymbol\theta}_2$ whose arclength is minimal among such curves.  The
length of such a minimal curve defines a distance function on $\mathcal M$.  In Euclidean space, where ${\boldsymbol G}_{0}({\boldsymbol\theta})\equiv
{\boldsymbol I}$, the shortest curve connecting ${\boldsymbol\theta}_1$ and ${\boldsymbol\theta}_2$ is simply a straight line with the Euclidean
length $\Vert {\boldsymbol\theta}_1-{\boldsymbol\theta}_2\Vert_2$.

As mentioned above, while standard HMC algorithms explore the target distribution more efficiently, they nevertheless fail to move between isolated
modes since these modes are separated by high energy barriers \citep{sminchisescu11}. To address this issue, we propose to replace the base metric
${\boldsymbol G}_{0}$ with a new metric for which the distance between modes is shortened. This way, we can facilitate moving between modes by
creating ``wormholes'' between them.

Let $\hat{\boldsymbol\theta}_1$ and $\hat{\boldsymbol\theta}_2$ be two modes of the target distribution. We define a straight line segment, ${\boldsymbol
v}_W:= \hat{\boldsymbol\theta}_2-\hat{\boldsymbol\theta}_1$, and refer to a small neighborhood (tube) of the line segment as a \emph{wormhole}. Next, we
define a \emph{wormhole metric}, ${\boldsymbol G_W}({\boldsymbol\theta})$, in the vicinity of the wormhole. The metric ${\boldsymbol
G_W}({\boldsymbol\theta})$ is an inner product assigning a non-negative real number to a pair of tangent vectors ${\boldsymbol u}, {\boldsymbol w}$:
${\boldsymbol G_W}({\boldsymbol\theta})({\boldsymbol u}, {\boldsymbol w})\in \mathbb R^+$. To shorten the distance in the direction of ${\boldsymbol v}_W$, we project both ${\boldsymbol u}, {\boldsymbol w}$ to the plane normal to ${\boldsymbol v}_W$ and then take the Euclidean inner product of
those projected vectors. We set ${\boldsymbol v}_W^*={\boldsymbol v}_W/\Vert {\boldsymbol v}_W\Vert$ and
define a \emph{pseudo wormhole metric} ${\boldsymbol G^*_W}$ as follows:
\begin{eqnarray*}\label{psuNwormhole}
{\boldsymbol G^*_W}({\boldsymbol u}, {\boldsymbol w}) & := & \langle {\boldsymbol u}-\langle {\boldsymbol u},{\boldsymbol v}_W^*\rangle {\boldsymbol v}_W^*, {\boldsymbol w}-\langle {\boldsymbol w},{\boldsymbol v}_W^*\rangle {\boldsymbol v}_W^* \rangle \\
 & = & {\boldsymbol u}^{\textsf T} [{\boldsymbol I} - {\boldsymbol v}_W^*({\boldsymbol v}_W^*)^{\textsf T}] {\boldsymbol w}
\end{eqnarray*}
Note that ${\boldsymbol G^*_W}:= {\boldsymbol I} - {\boldsymbol v}_W^*({\boldsymbol v}_W^*)^{\textsf T}$ is semi-positive
definite (degenerate at ${\boldsymbol v}_W^*\neq 0$). We modify this metric to make it positive definite, and define the \emph{wormhole
metric} ${\boldsymbol G_W}$ as follows:
\begin{equation}\label{Nwormhole}
{\boldsymbol G_W} = {\boldsymbol G^*_W} +\eps {\boldsymbol v}_W^*({\boldsymbol v}_W^*)^{\textsf T} =  {\boldsymbol I} - (1-\eps){\boldsymbol v}_W^*({\boldsymbol v}_W^*)^{\textsf T}
\end{equation}
where $0 < \epsilon \ll 1$ is a small positive number.

To see that the wormhole metric ${\boldsymbol G_W}$ in fact shortens the distance between $\hat{\boldsymbol\theta}_1$ and $\hat{\boldsymbol\theta}_2$,
consider a simple case where ${\boldsymbol\theta}(t)$ follows a straight line: ${\boldsymbol\theta}(t) = {\boldsymbol\theta}_1+ {\boldsymbol v}_Wt,
t\in [0,1]$.
In this case, the distance under ${\boldsymbol G_W}$ is 
\begin{equation*}
  \mathrm{dist}(\hat{\boldsymbol\theta}_1, \hat{\boldsymbol\theta}_2) =
  \int_0^1 \sqrt{ {\boldsymbol v}_W^{\textsf T} {\boldsymbol G_W} {\boldsymbol v}_W }dt
  = \sqrt\eps \Vert {\boldsymbol v}_W\Vert \ll \Vert {\boldsymbol v}_W\Vert
\end{equation*}
which is much smaller than the Euclidean distance.

Next, we define the overall metric, $\boldsymbol G$, for the whole parameter space of $\boldsymbol \theta$ as a weighted sum of the base metric
${\boldsymbol G_0}$ and the wormhole metric ${\boldsymbol G_W}$,
\begin{equation}\label{finalM}
{\boldsymbol G}({\boldsymbol\theta}) = (1-\mathfrak m({\boldsymbol\theta})) {\boldsymbol G_0}({\boldsymbol\theta}) + \mathfrak m({\boldsymbol\theta}) {\boldsymbol G_W}
\end{equation}
where $\mathfrak m({\boldsymbol\theta})\in(0,1)$ is a mollifying function designed to make the wormhole metric ${\boldsymbol G_W}$ influential in the
vicinity of the wormhole only. In this paper, we choose the following mollifier:
\begin{equation}\label{mollNT}
\mathfrak m({\boldsymbol\theta}) := \exp \{ -(\Vert {\boldsymbol\theta} - \hat{\boldsymbol\theta}_1\Vert + \Vert{\boldsymbol\theta} - \hat{\boldsymbol\theta}_2\Vert - \Vert\hat{\boldsymbol\theta}_1 - \hat{\boldsymbol\theta}_2\Vert) /F \}
\end{equation}
where the \emph{influence factor} $F>0$, is a free parameter that can be tuned to modify the extent of the influence of ${\boldsymbol G_W}$:
decreasing $F$ makes the influence of ${\boldsymbol G_W}$ more restricted around the wormhole. The resulting metric leaves the base metric almost intact
outside of the wormhole, while making the transition of the metric from outside to inside smooth. Within the wormhole, the trajectories are mainly guided
in the wormhole direction ${\boldsymbol v}_W^*$: ${\boldsymbol G}({\boldsymbol\theta}) \approx {\boldsymbol G_W}$, so ${\boldsymbol
G}({\boldsymbol\theta})^{-1} \approx {\boldsymbol G_W}^{-1}$ has the dominant eigen-vector ${\boldsymbol v}_W^*$ (with eigen-value
$1/\eps\gg 1$), thereafter ${\boldsymbol v}\sim \mathcal N({\bf 0},{\boldsymbol G}({\boldsymbol\theta})^{-1})$ tends to be directed in ${\boldsymbol v}_W^*$.

\begin{figure}[t]
\begin{center}
 \centerline{
\includegraphics[width=5.5in, height=2in]{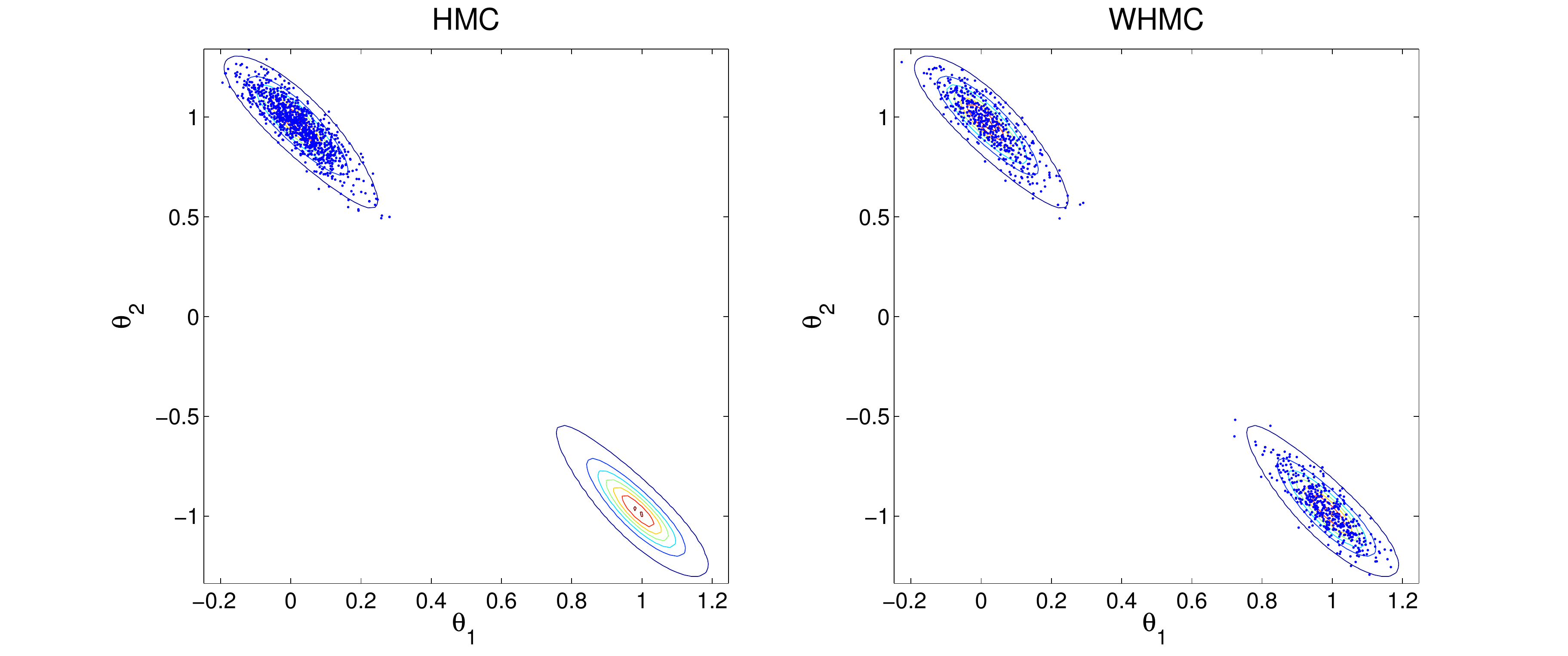}
}
\caption{Comparing HMC and WHMC in terms of sampling from a two-dimensional posterior distribution with two isolated modes.}
\vspace{-18pt}
\label{2mode}
\end{center}
 \end{figure}


We use the modified metric \eqref{finalM} in RMHMC and refer to the resulting algorithm as Wormhole Hamiltonian Monte Carlo (WHMC). Figure \ref{2mode} compares WHMC to standard HMC based on the following illustrative example appeared in the paper by \cite{welling11}:
\begin{eqnarray*}
\theta_d & \sim & \mathcal N(\theta_d, \sigma_d^2), \quad d=1,2.\\
x_i & \sim & \frac{1}{2} \mathcal N(\theta_1, \sigma_x^2) + \frac{1}{2} \mathcal N(\theta_1+\theta_2, \sigma_x^2).
\end{eqnarray*}
Here, we set $\theta_1=0, \theta_2=1, \sigma_1^2=10, \sigma_2^2=1$, $\sigma_x^2=2$, and generate 1000 data points from the above model. In Figure \ref{2mode}, the dots show the posterior samples of $(\theta_1, \theta_2)$ given the simulated data. While HMC is trapped in one mode, WHMC moves easily between the two modes. For this example, we set ${\boldsymbol G_0} = \boldsymbol I$ to make WHMC comparable to standard HMC. Further, we use $0.03$ and $0.3$ for $\eps$ and $F$ respectively. 

For more than two modes, we can construct a network of wormholes by connecting any two modes with a wormhole. Alternatively, we can create a wormhole between neighboring modes only. In this paper, we define the neighborhood using a \emph{minimal spanning tree} \citep{kleinberg05}.

The above method could suffer from two potential shortcomings in higher dimensions. First, the effect of wormhole metric could diminish fast as the sampler leaves one mode towards another mode. Second, such mechanism, which modifies the
dynamics in the existing parameter space, could interfere with the native HMC dynamics in the neighborhood of a wormhole.

To address the first issue, we add an external vector field to enforce the movement between modes. More specifically, we define a vector field, ${\boldsymbol
f}({\boldsymbol\theta},{\boldsymbol v})$, in terms of the position parameter $\boldsymbol \theta$ and the velocity vector ${\boldsymbol
v}={\boldsymbol G}({\boldsymbol\theta})^{-1}{\boldsymbol p}$ as follows:
\begin{eqnarray*}\label{xf}
{\boldsymbol f}({\boldsymbol\theta}, {\boldsymbol v}) & := & \exp\{-V({\boldsymbol\theta})/(DF)\}U({\boldsymbol\theta}) \langle {\boldsymbol v},{\boldsymbol v}_W^*\rangle {\boldsymbol v}_W^* \\
& = &  \mathfrak m({\boldsymbol\theta}) \langle {\boldsymbol v},{\boldsymbol v}_W^*\rangle {\boldsymbol v}_W^*
\end{eqnarray*}
with mollifier $\mathfrak m({\boldsymbol\theta}):= \exp\{-V({\boldsymbol\theta})/(DF)\}$, where $D$ is the dimension, $F>0$ is the influence factor,
and $V({\boldsymbol\theta})$ is a vicinity function indicating the Euclidean distance from the line segment ${\boldsymbol v}_W$,
\begin{equation}\label{vicinity}
V({\boldsymbol\theta}):= \langle {\boldsymbol\theta}-\hat{\boldsymbol\theta}_1,{\boldsymbol\theta}-\hat{\boldsymbol\theta}_2\rangle + |\langle {\boldsymbol\theta}-\hat{\boldsymbol\theta}_1,{\boldsymbol v}_W^*\rangle||\langle {\boldsymbol\theta}-\hat{\boldsymbol\theta}_2,{\boldsymbol v}_W^*\rangle|
\end{equation}
The resulting vector field has three properties: 1) it is confined to a neighborhood of each wormhole, 2) it enforces the movement along the wormhole, and 3)
its influence diminishes at the end of the wormhole when the sampler reaches another mode.

\begin{figure}[t]
\begin{center}
 \centerline{
\includegraphics[width=5.5in, height=2in]{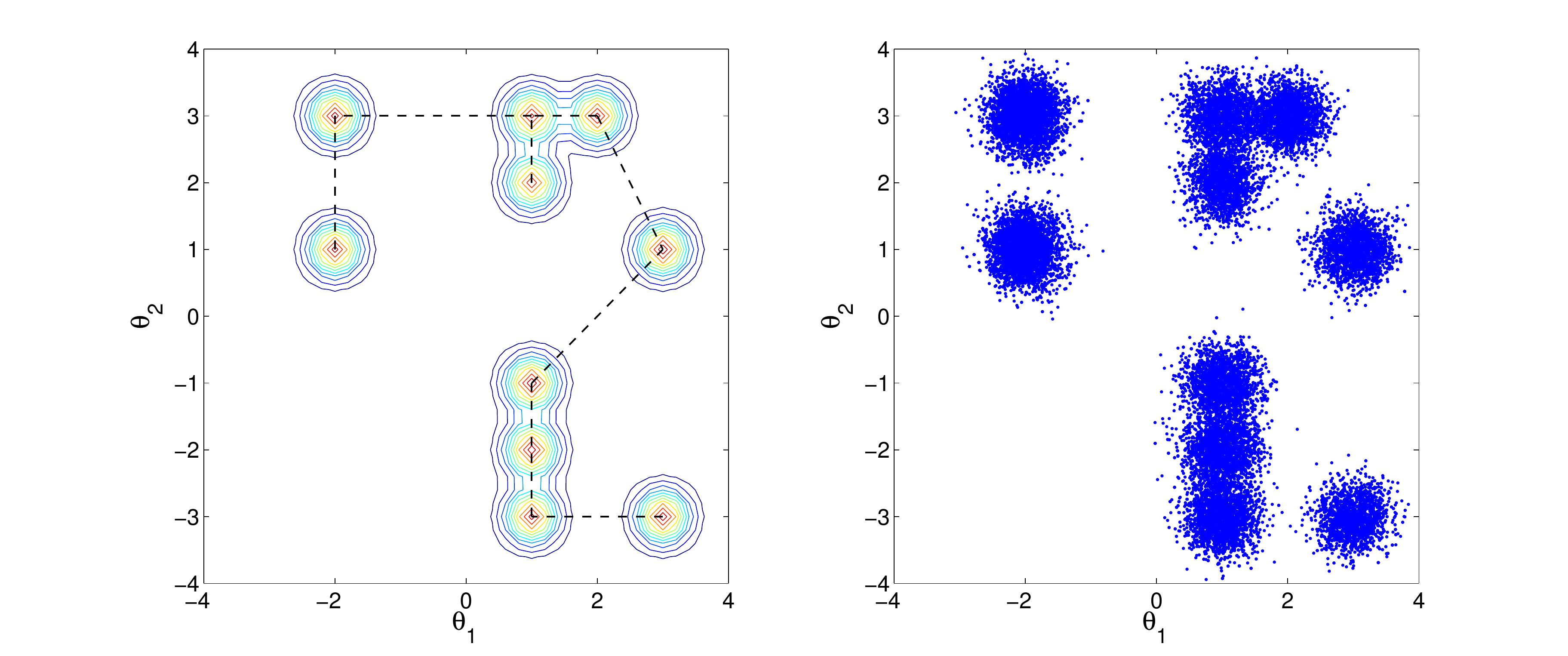}
}
\caption{Sampling from a mixture of 10 Gaussian distributions with dimension $D=100$ using WHMC with a vector field ${\boldsymbol
f}({\boldsymbol\theta}, {\boldsymbol v})$ to enforce moving between modes in higher dimensions.}
\vspace{-18pt}
\label{D100K10}
\end{center}
 \end{figure}

After adding the vector field, we modify the Hamiltonian equation governing the evolution of ${\boldsymbol\theta}$ as follows:
\begin{equation}\label{windeq}
\dot {\boldsymbol\theta} = {\boldsymbol v} + {\boldsymbol f}({\boldsymbol\theta}, {\boldsymbol v})
\end{equation}
We also need to adjust the Metropolis acceptance probability accordingly since the transformation is not volume preserving. (More details are provided in the supplementary file.) Figure \ref{D100K10} illustrates this approach based on sampling from a mixture of 10 Gaussian distributions with dimension $D=100$.

To address the second issue, we allow the wormholes to pass through
an extra auxiliary dimension to avoid their interference with the existing HMC dynamics in the given parameter space. In
particular we introduce an auxiliary variable $\theta_{D+1}\sim \mathcal N(0,1)$ corresponding to an auxiliary dimension. We use
$\tilde{\boldsymbol\theta} := ({\boldsymbol\theta},\theta_{D+1})$ to denote the position parameters in the resulting $D+1$ dimensional space $\mathcal
M^{D}\times \mathbb R$. $\theta_{D+1}$ can be viewed as random noise independent of $\boldsymbol\theta$ and contributes $\frac{1}{2}\theta_{D+1}^2$ to the total potential energy. Correspondingly, we augment velocity ${\boldsymbol v}$ with one extra dimension, denoted as $\tilde{\boldsymbol v}:=({\boldsymbol v},v_{D+1})$. At the end of the sampling, we project $\tilde{\boldsymbol\theta}$ to the original parameter space and discard $\theta_{D+1}$.

We refer to $\mathcal M^{D}\times \{-h\}$ as the \emph{real world}, and call $\mathcal M^{D}\times \{+h\}$ the \emph{mirror world}. Here, $h$ is half of the distance between the two worlds, and it should be in the same scale as the average distance between the modes. For most of the examples discussed here, we set $h=1$. Figure \ref{whpic} illustrates how the two worlds
are connected by networks of wormholes. When the sampler is near a mode $(\hat{\boldsymbol\theta}_1, -h)$ in the real world, we build a wormhole network by connecting it to all the modes in the mirror world. Similarly, we connect the corresponding mode in the mirror world, $(\hat{\boldsymbol\theta}_1, +h)$, to all the modes in the real world. Such construction allows the sampler to jump from one mode in the real world to the same mode in the mirror world and vice versa. This way, the algorithm can effectively sample from the vicinity of a mode, while occasionally jumping from one mode to another.

\begin{figure}[t]
  \begin{center}
    \includegraphics[scale=0.5]{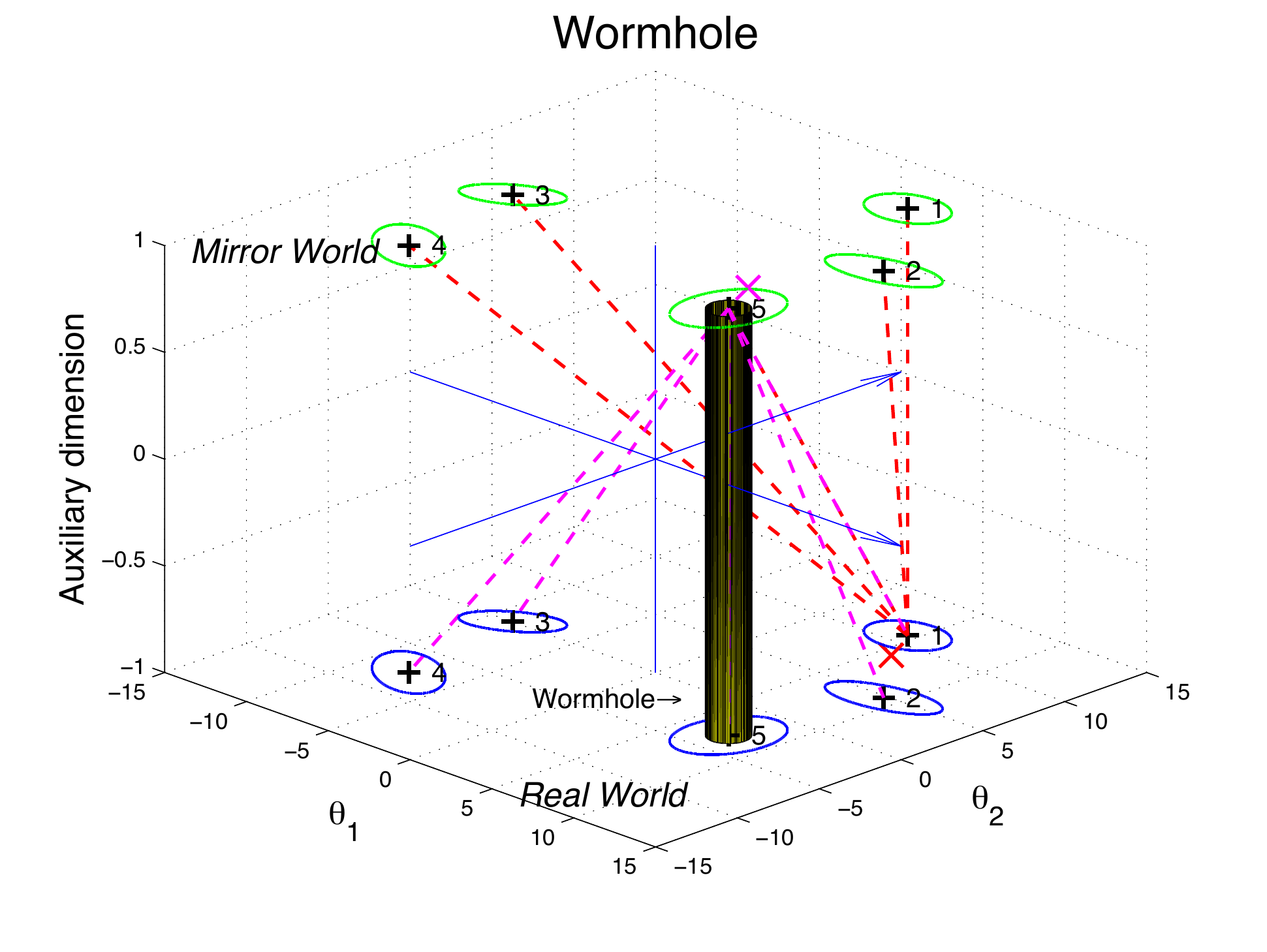}
  \end{center}
  \vspace{-18pt}
  \caption{Illustrating a wormhole network connecting the real world to the mirror world ($h=1$). As an example, the cylinder shows a wormhole
  connecting mode 5 in the real world to its mirror image. The dashed lines show two sets of wormholes. The red lines shows the wormholes when the sampler is close to mode 1 in the real world, and the magenta lines show the wormholes when the sampler is close to mode 5 in the mirror world.}
  \label{whpic}
 \vspace{5pt}
\end{figure}

The attached supplementary file provides the details of our algorithm (Algorithm 1), along with the proof of convergence and its implementation in MATLAB. 

\section{Mode Searching After Regeneration}
So far, we assumed that the locations of modes are known. This is of course not a realistic assumption in many situations. In this section, we relax
this assumption by extending our method to search for new modes proactively and to update the network of wormholes dynamically. In general, however,
allowing such adaptation to take place infinitely often will disturb the stationary distribution of the chain, rendering the process no longer Markov
\citep{gelfand94, gilks98}. To avoid this issue, we use the \emph{regeneration} method discussed by \cite{nummelin84, mykland95, gilks98,
brockwell05}.

Informally, a regenerative process ``starts again'' probabilistically at a set of times, called \emph{regeneration times}
\citep{brockwell05}. At regeneration times, the transition mechanism can be modified based on the entire history of
the chain up to that point without disturbing the consistency of MCMC estimators. 


\subsection{Identifying Regeneration Times}
The main idea behind finding regeneration times is to regard the transition kernel $T(\boldsymbol\theta_{t+1}|\boldsymbol\theta_{t})$ as a mixture of two kernels, $Q$ and $R$ \citep{nummelin84,ahn13},
\begin{equation*}
T(\boldsymbol\theta_{t+1}|\boldsymbol\theta_{t})= S(\boldsymbol\theta_{t})Q(\boldsymbol\theta_{t+1}) +
(1-S(\boldsymbol\theta_{t}))R(\boldsymbol\theta_{t+1}|\boldsymbol\theta_{t})
\end{equation*}
where $Q({\boldsymbol\theta}_{t+1})$ is an \emph{independence kernel}, and the \emph{residual kernel} $R(\boldsymbol\theta_{t+1}|\boldsymbol\theta_{t})$ is defined
as follows:
\begin{equation*}
\!\!R(\boldsymbol\theta_{t+1}|\boldsymbol\theta_{t})\!\!=\!\!
\begin{dcases}
\frac{T(\boldsymbol\theta_{t+1}|\boldsymbol\theta_{t})-S(\boldsymbol\theta_{t})Q(\boldsymbol\theta_{t+1})}{1-S(\boldsymbol\theta_{t})}, &\!\!\!\!
S(\boldsymbol\theta_{t})\in [0,1)\\
1, &\!\!\!\! S(\boldsymbol\theta_{t})=1
\end{dcases}
\end{equation*}
$S(\boldsymbol\theta_{t})$ is the mixing coefficient between the two kernels such that
\begin{equation}\label{Kerineq}
T(\boldsymbol\theta_{t+1}|\boldsymbol\theta_{t}) \ge
S(\boldsymbol\theta_{t})Q(\boldsymbol\theta_{t+1}), \forall \boldsymbol\theta_{t}, \boldsymbol\theta_{t+1}
\end{equation}

Now suppose that at iteration $t$, the current state is $\boldsymbol\theta_{t}$. 
To implement this approach, we first generate $\boldsymbol\theta_{t+1}$ using the original transition kernel
$\boldsymbol\theta_{t+1}|\boldsymbol\theta_{t}\sim T(\cdot|\boldsymbol\theta_{t})$. Then, we sample $B_{t+1}$ from a Bernoulli distribution with
probability
\begin{equation}\label{retroProb}
r(\boldsymbol\theta_{t}, \boldsymbol\theta_{t+1}) = \frac{S(\boldsymbol\theta_{t})Q(\boldsymbol\theta_{t+1})}{T(\boldsymbol\theta_{t+1}|\boldsymbol\theta_{t})}
\end{equation}
If $B_{t+1} = 1$, a regeneration has occurred, then we discard $\boldsymbol\theta_{t+1}$ and sample it from the independence kernel
$\boldsymbol\theta_{t+1}\sim Q(\cdot)$.
At regeneration times, we redefine the dynamics using the past sample path. 

Ideally, we would like to evaluate regeneration times in terms of WHMC's transition kernel. In general, however, this is quite difficult for such Metropolis algorithms. On the other hand, regenerations are easily achieved for the independence sampler (i.e., the proposed state is independent from the current state) as long as the proposal distribution is close to the target distribution \citep{gilks98}. Therefore, we can specify a hybrid sampler that consists of the original proposal distribution (here, WHMC) and the independence sampler, and adapt both proposal distributions whenever a regeneration is obtained on an independence-sampler step \citep{gilks98}. In our method, we systematically alternate between WHMC and the independence sampler while evaluating regeneration times based on the independence sampler only. 

As mentioned above, for this method to be effective, the proposal distribution for the independence sampler should be close to the target distribution. To this end, we follow \cite{ahn13} and specify our independence sampler as a mixture of Gaussians located at the previously identified modes. The covariance matrix for each mixture component is set to the inverse observed Fisher information (i.e., Hessian) evaluated at the mode. The relative weight of each mixture component set to $1/k$ at the beginning, where $k$ is the number of previously identified modes through some initial optimization algorithm. The weights are updated at regeneration times to be proportional to the number of times the corresponding mode has been visited. More specifically, $T(\boldsymbol\theta_{t+1}|\boldsymbol\theta_{t})$, $S(\boldsymbol\theta_{t})$ and $Q(\boldsymbol\theta_{t+1})$ are defined as
follows to satisfy \eqref{Kerineq}:
\begin{align}
T(\boldsymbol\theta_{t+1}|\boldsymbol\theta_{t}) &= q(\boldsymbol\theta_{t+1})
\min\left\{1,\frac{\pi(\boldsymbol\theta_{t+1})/q(\boldsymbol\theta_{t+1})}{\pi(\boldsymbol\theta_{t})/q(\boldsymbol\theta_{t})}\right\}
\label{TSQ:T}\\
S(\boldsymbol\theta_{t}) &= \min\left\{1,\frac{c}{\pi(\boldsymbol\theta_{t})/q(\boldsymbol\theta_{t})}\right\} \label{TSQ:S}\\
Q(\boldsymbol\theta_{t+1}) & = q(\boldsymbol\theta_{t+1}) \min\left\{1,\frac{\pi(\boldsymbol\theta_{t+1})/q(\boldsymbol\theta_{t+1})}{c}\right\}
\label{TSQ:Q}
\end{align}
where $q(\cdot)$ is the independence proposal kernel, which specified using a mixture of Gaussians with means fixed at the $k$ known modes prior to regeneration. Algorithm 2 in the supplementary file shows the steps for this method.

\begin{figure}[t]
 \centerline{ 
\includegraphics[height=2in, width=2in]{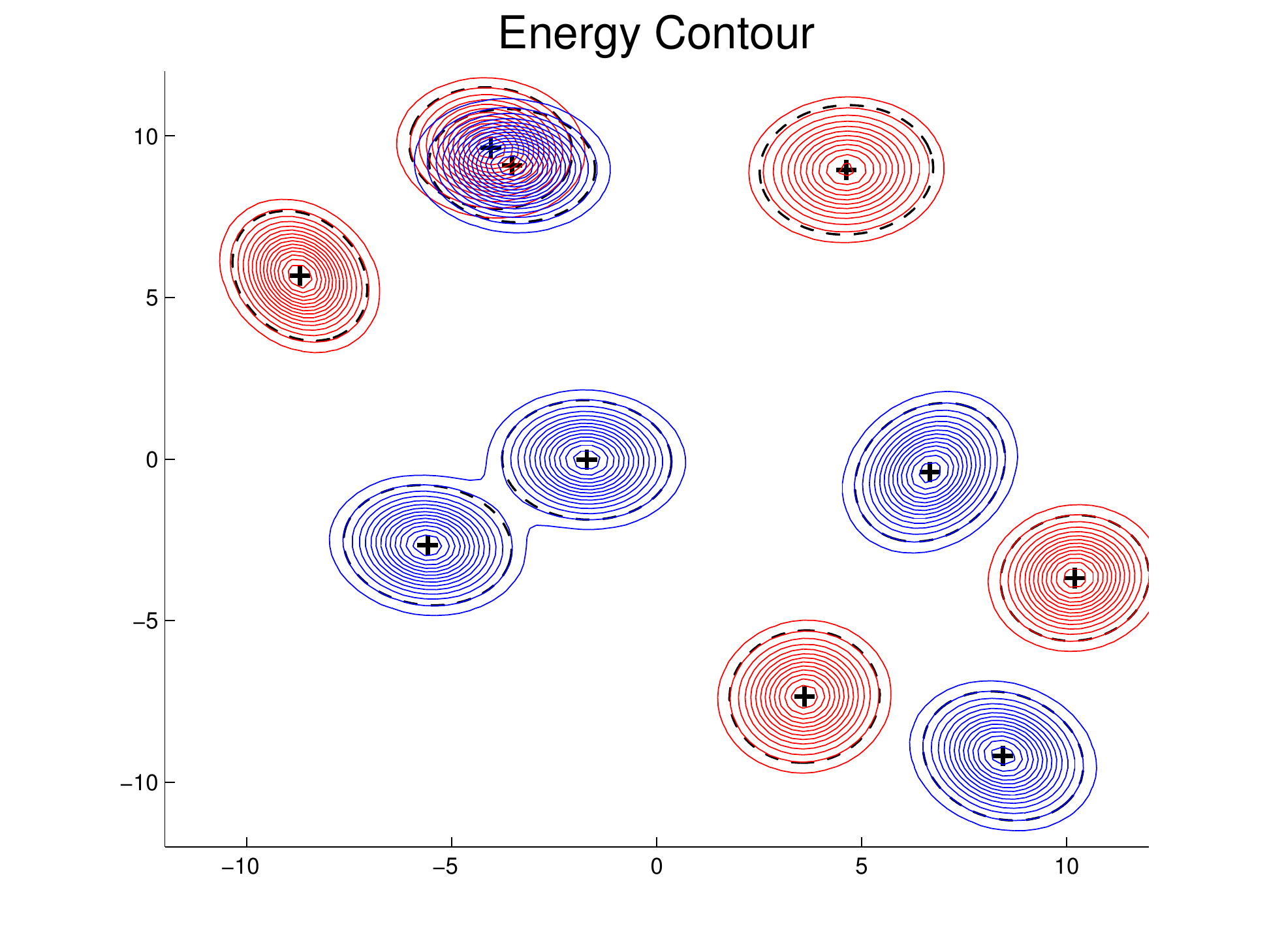}
\includegraphics[height=2in, width=2in]{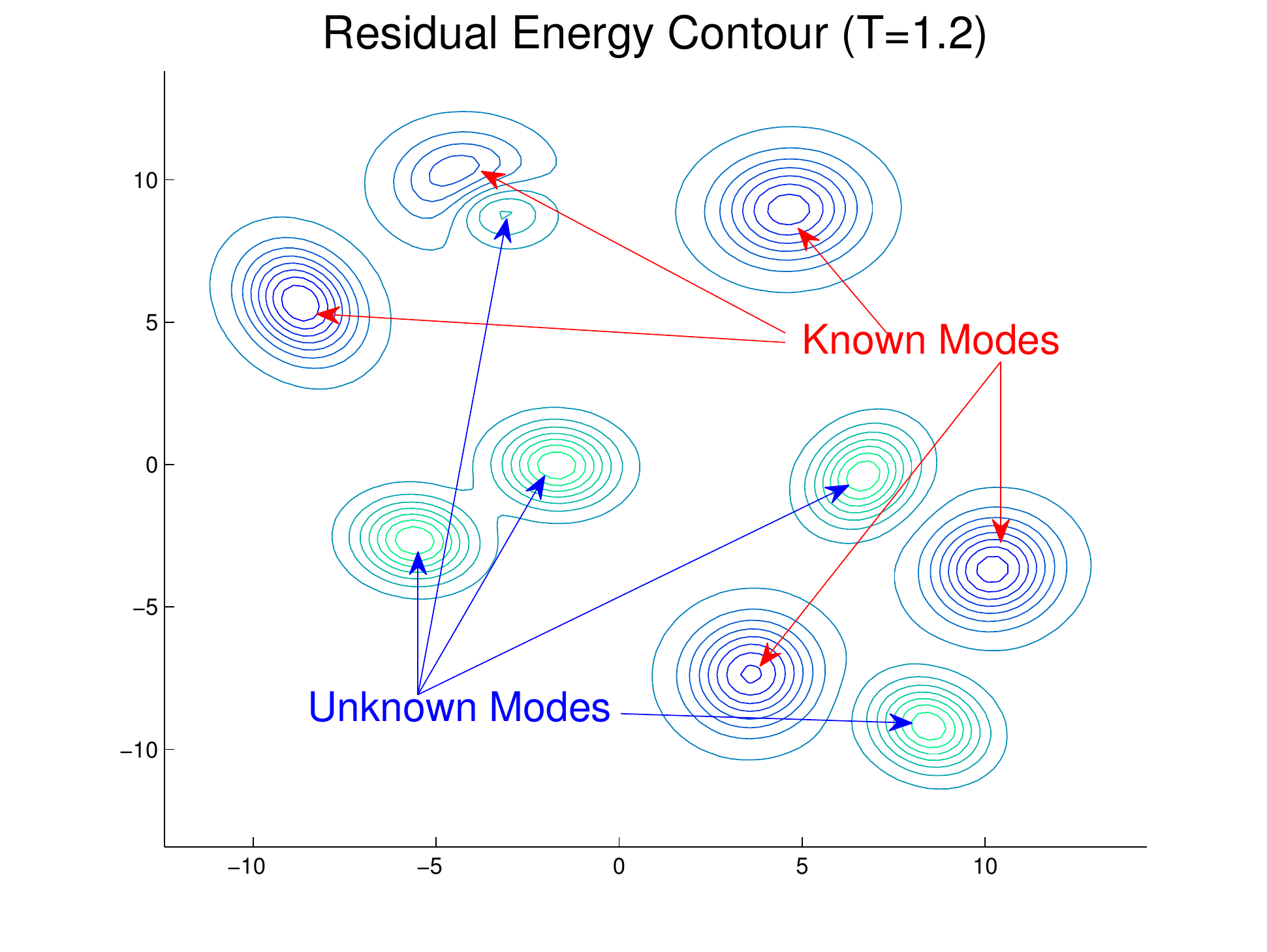}
\includegraphics[height=2in, width=2in]{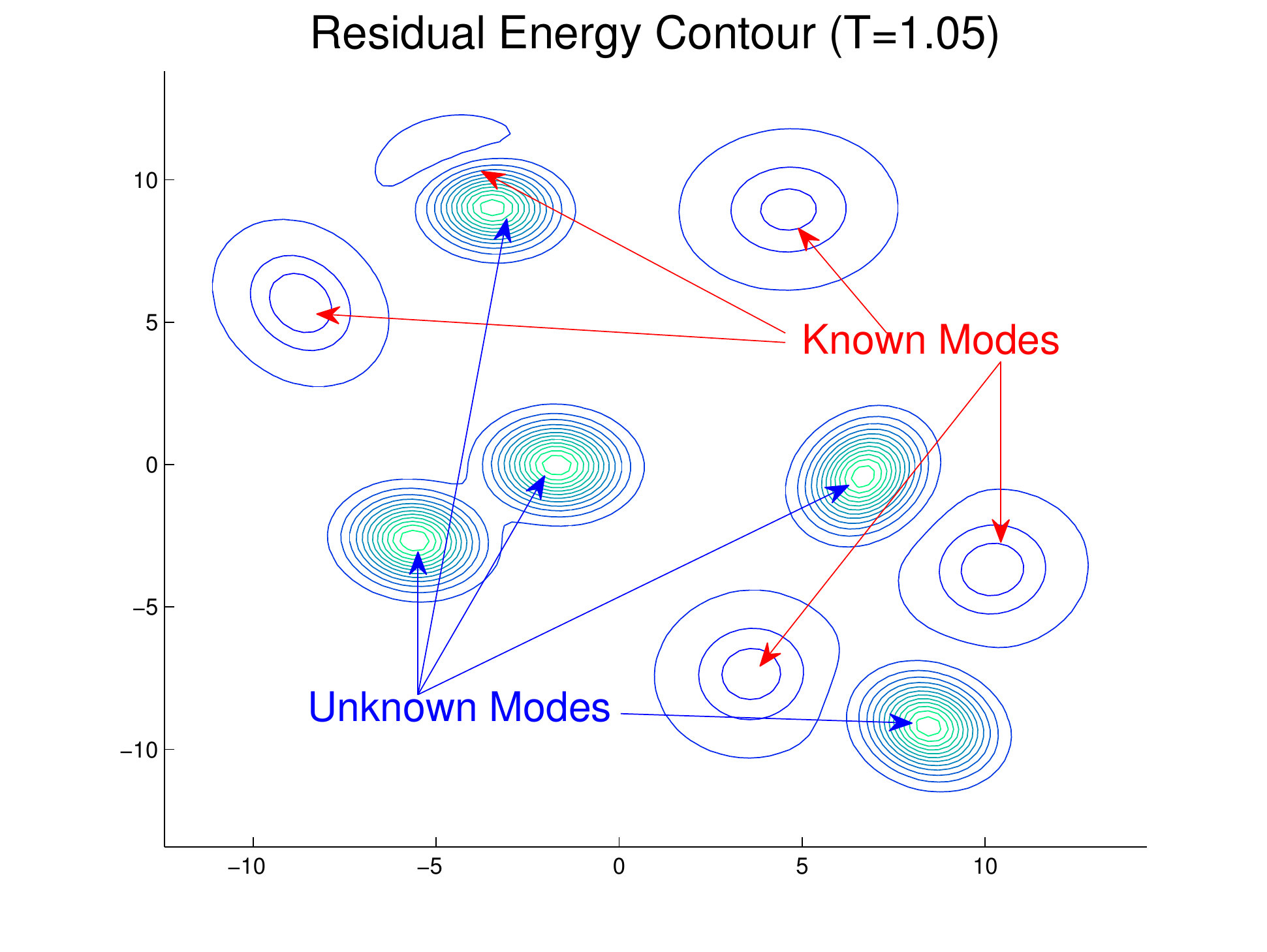}
}
\caption{Left panel: True energy function (red: known modes, blue: unknown modes). Middle panel: Residual energy function at $T=1.2$. Right panel: Residual energy function at $T=1.05$.}
\label{engysurgT}
\end{figure}

\subsection{Identifying New Modes}
When the chain regenerates, we can search for new modes, modify the transition kernel by including newly found modes in the mode library, and update the wormhole network accordingly. This way, starting with a limited number of modes (identified by some preliminary optimization method), WHMC could discover unknown modes on the fly without affecting the stationarity of the chain.

To search for new modes after regeneration, we could simply apply an optimization algorithm to the original target density function $\pi(\boldsymbol\theta)$ with some random starting points. This, however, could lead to frequently rediscovering the known modes. To avoid this issue, we propose to remove/down-weight the known modes using the history of the chain up to the regeneration time and run an optimization
algorithm on the resulting \emph{residual density}, or equivalently, on the corresponding \emph{residual energy} (i.e., minus log of density). To this end, we fit a mixture of Gaussians with the best knowledge of modes (locations, Hessians and relative weights) prior to the regeneration. The \emph{residual density} function could be simply defined as $\pi_{\bf r}(\boldsymbol\theta) = \pi(\boldsymbol\theta)-q(\boldsymbol\theta)$ with
the corresponding \emph{residual potential energy} as follows,
\begin{equation*}
U_{\bf r}(\boldsymbol\theta) = \log(\pi_{\bf r}(\boldsymbol\theta)+c)= -\log(\pi(\boldsymbol\theta)-q(\boldsymbol\theta)+c)
\end{equation*}
where the constant $c>0$ is used to make the term inside the log function positive. To avoid completely flat regions (e.g., when a Gaussian distribution provides a good approximation around the mode), which could cause gradient-based optimization methods to fail, we could use the following \emph{tempered residual potential energy} instead:
\begin{equation*}\label{tempresU}
U_{\bf r}(\boldsymbol\theta,T) = -\log \left(\pi(\boldsymbol\theta)-\exp\left(\frac{1}{T}\log q(\boldsymbol\theta)\right)+c\right)
\end{equation*}
where $T$ is the temperature. Figure \ref{engysurgT} illustrates this concept. 

When the optimizer finds new modes, they are added to the existing
mode library, and the wormhole network is updated accordingly.

\begin{figure}[t]
\begin{center}
 \centerline{
\includegraphics[height=2in, width=2.4in]{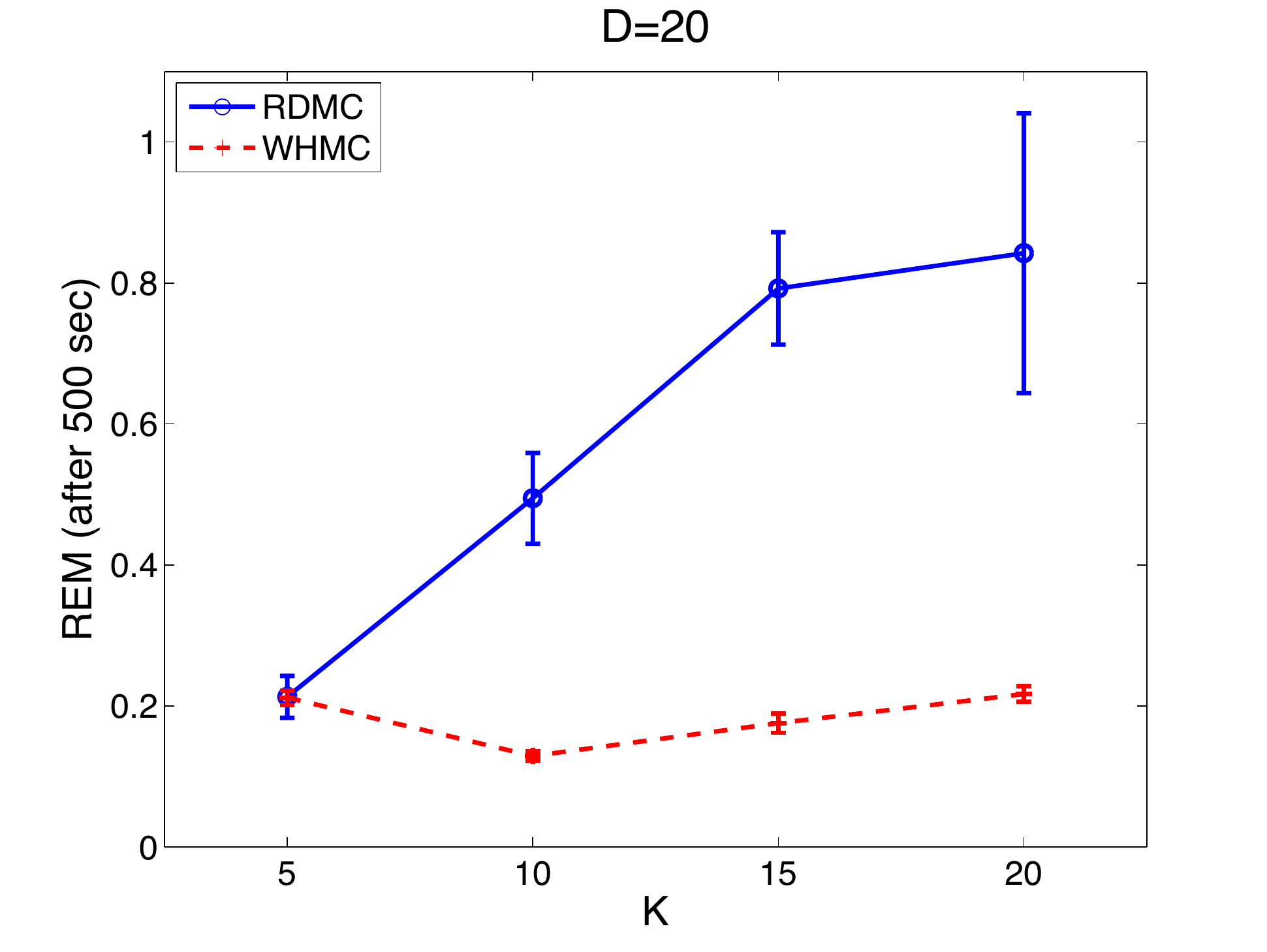}
\includegraphics[height=2in, width=2.4in]{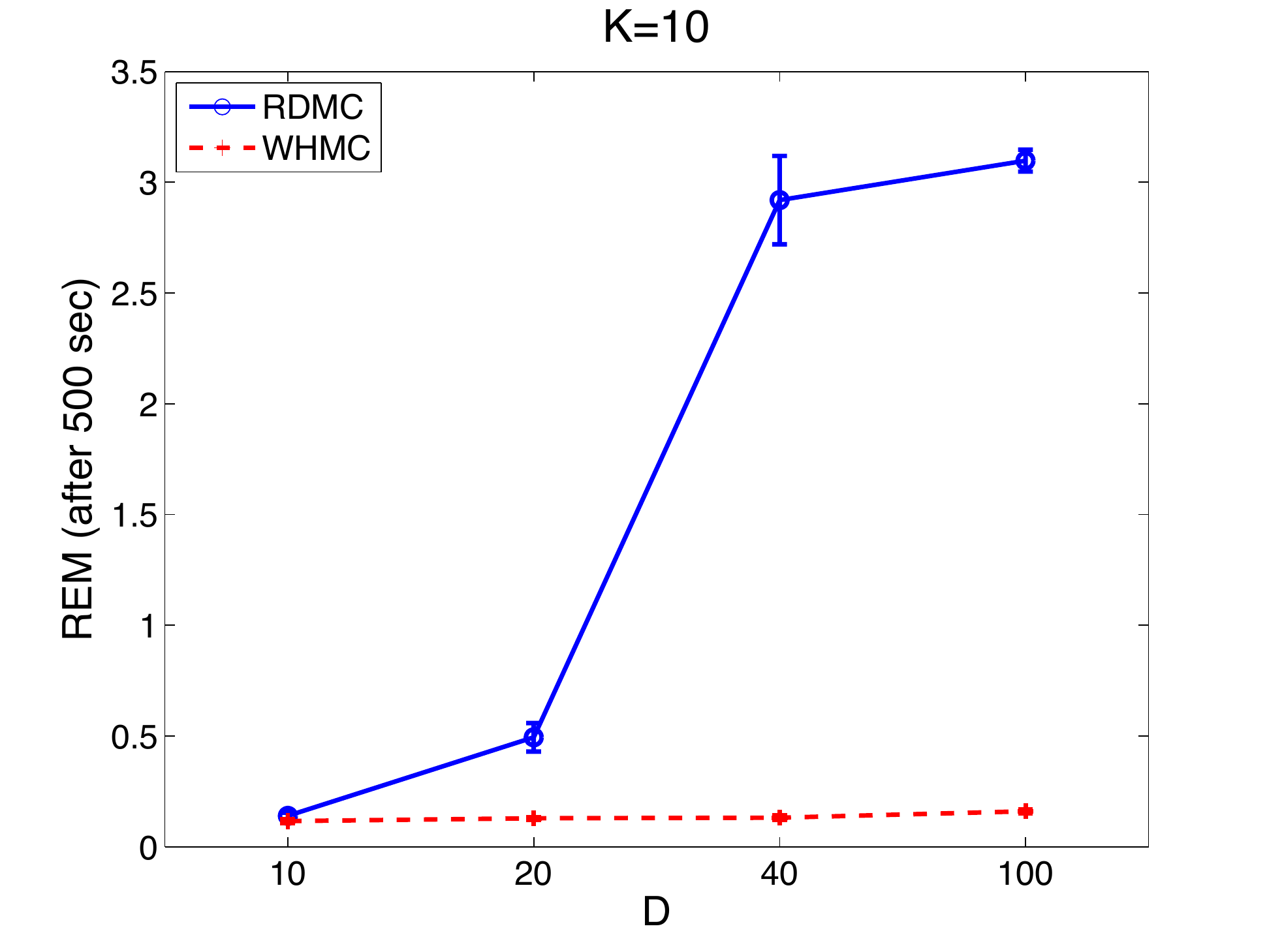}
}
\caption{Comparing WHMC to RDMC using $K$ mixtures of $D$-dimensional Gaussians. Left panel: REM (along with 95\% confidence interval based on 10 MCMC
chains) for varying number of mixture components, $K=5, 10, 15, 20$, with fixed dimension, $D=20$. Right panel: REM (along with 95\% confidence
interval based on 10 MCMC chains) for varying number of dimensions, $D=10, 20, 40, 100$, with fixed number of mixture components, $K=10$.}
\vspace{-18pt}
\label{KDplot}
\end{center}
\end{figure}

\section{Empirical Results} \label{experiments}

In this section, we evaluate the performance of our method, henceforth called Wormhole Hamiltonian Monte Carlo (WHMC), using three examples. The first example involves
sampling from mixtures of Gaussian distributions with varying number of modes and dimensions. In this example, which is also discussed by \cite{ahn13}, the locations of modes are assumed to be known. The second
example, which was originally proposed by \cite{Ihler05}, involves inference regarding the locations of sensors in a network. For our third example, we also use mixtures of Gaussian distributions, but this time we assume that the
locations of modes are unknown.

We evaluate our method's performance by comparing it to Regeneration Darting Monte Carlo (RDMC) \citep{ahn13}, which is one of the most recent
algorithms designed for sampling from multimodal distributions based on the Darting Monte Carlo (DMC) \citep{sminchisescu11} approach. DMC defines
high density regions around the modes. When the sampler enters these regions, a jump between the regions will be attempted. RDMC enriches the DMC
method by using the regeneration approach \citep{mykland95, gilks98}.

We compare the two methods (i.e., WHMC and RDMC) in terms of Relative Error of Mean (REM) \citep{ahn13}, which summarizes the errors in approximating
the expectation of variables across all dimensions, and its value at time $t$ is $\mathrm{REM}(t) = \Vert \bar{\theta(t)}-\theta^{*}\Vert_{1}/ \Vert
\theta^{*}\Vert_{1}$. Here, $\bar{\theta(t)}$ is the mean of MCMC samples at time $t$ and $\theta^{*}$ is the true mean. Because RDMC uses standard
HMC algorithm with a flat metric, we also use the baseline metric ${\boldsymbol G}_{0} \equiv{\boldsymbol I}$ to make the two algorithms comparable. Our
approach, however, can be easily extended to other metrics such as Fisher information.

\begin{figure}[t]
\begin{center}
 \centerline{ 
 \includegraphics[height=2in, width=5.5in]{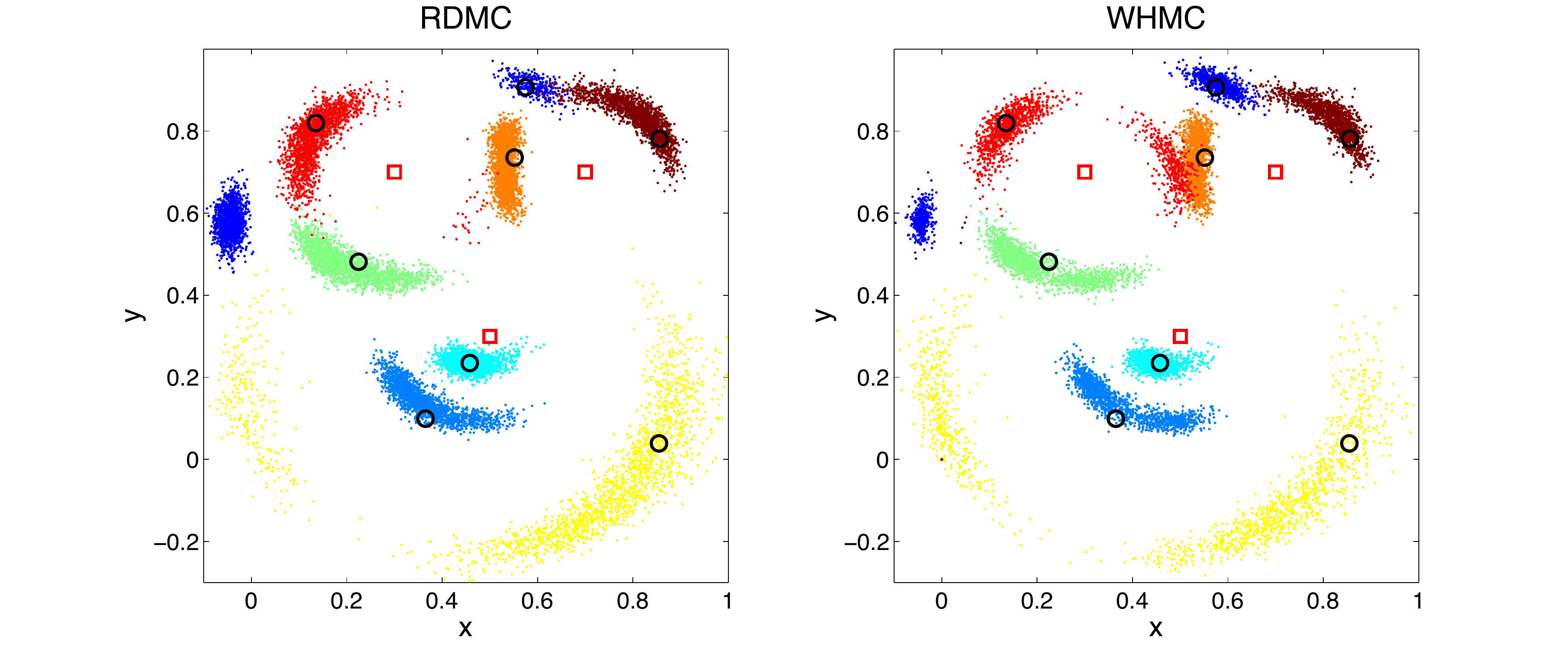}
}
\caption{Posterior samples for sensor locations using RDMC (left panel) and WHMC (right panel).}
\vspace{-18pt}
\label{timeplotwsn}
\end{center}
\end{figure}

\subsection{Mixture of Gaussians with Known Modes}
First, we evaluate the performance of our method based on sampling from $K$ mixtures of $D$-dimensional Gaussian distributions with \emph{known}
modes. (We relax this assumption later.) The means of these distributions are randomly generated from $D$-dimensional uniform
distributions such that the average pairwise distances remains around 20. The corresponding covariance matrices are constructed in a way that mixture
components have different density functions. Simulating samples from the resulting $D$ dimensional mixture of $K$ Gaussians is challenging because the
modes are far apart and the high density regions have different shapes.

The left panel of Figure \ref{KDplot} compares the two methods for varying number of mixture components, $K=5, 10, 15, 20$, with fixed dimension ($D=20$). The right panel
shows the results for varying number of dimensions, $D=10, 20, 40, 100$, with fixed number of mixture components ($K=10$). For both scenarios, we stop the two algorithms
after 500 seconds and compare their REM. As we can see, WHMC has substantially lower REM compared to RDMC, especially when the number of modes and
dimensions increase.


\subsection{Sensor Network Localization}
For our second example, we use a problem previously discussed by \cite{Ihler05} and \cite{ahn13}. We assume that $N$ sensors are scattered in a planar region with $2d$ locations
denoted as $\{x_{i}\}_{i=1}^{N}$. The distance $Y_{ij}$ between a pair of sensors $(x_{i}, x_{j})$ is observed with probability $\pi (x_{i},x_{j}) =
\exp(-\Vert x_{i}-x_{j}\Vert^{2}/(2R^{2}))$. If the distance is in fact observed ($Y_{ij}>0$), then $Y_{ij}$ follows a Gaussian distribution $\mathcal
N(\Vert x_{i}-x_{j}\Vert, \sigma^{2})$ with small $\sigma$; otherwise $Y_{ij}=0$. That is,
\begin{align*}
Z_{ij}=I(Y_{ij}>0)|x & \sim \mathrm{Binom}(1,\pi (x_{i},x_{j}))\\
Y_{ij}|Z_{ij}=1,x & \sim \mathcal N(\Vert x_{i}-x_{j}\Vert, \sigma^{2})
\end{align*}
where $Z_{ij}$ is a binary indicator set to 1 if the distance between $x_{i}$ and $x_{j}$ is observed. 

Given a set of observations $Y_{ij}$ and prior distribution of $x$, which is assumed to be uniform in this example, it is of interest to infer the
posterior distribution of all the sensor locations. Following \cite{ahn13}, we set $N=8, R=0.3, \sigma=0.02$, and add three additional base sensors
with known locations to avoid ambiguities of translation, rotation, and negation (mirror symmetry). The location of the 8 sensors form a multimodal
distribution with dimension $D=16$.

Figure \ref{timeplotwsn} shows the posterior samples based on the two methods. As we can see, RDMC very rarely visits one of the modes (shown in red
in the top middle part); whereas, WHMC generates enough samples from this mode to make it discernible. As a result, WHMC converges to a substantially lower REM ($0.02$)
compared to RDMC ($0.13$) after 500 seconds. 


\begin{figure}[t]
\begin{center}
 \centerline{
\includegraphics[height=2in, width=2.4in]{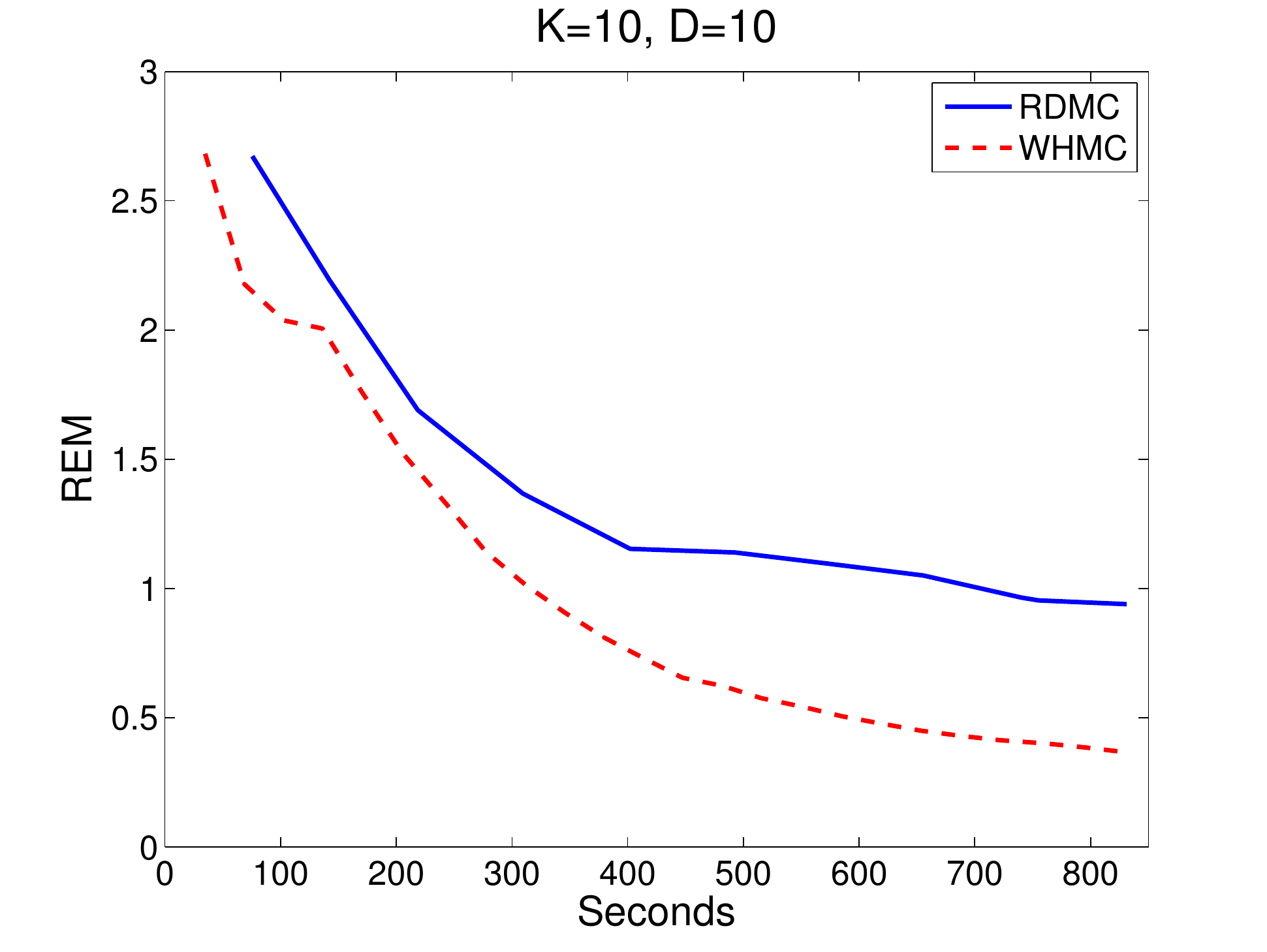}
\includegraphics[height=2in, width=2.4in]{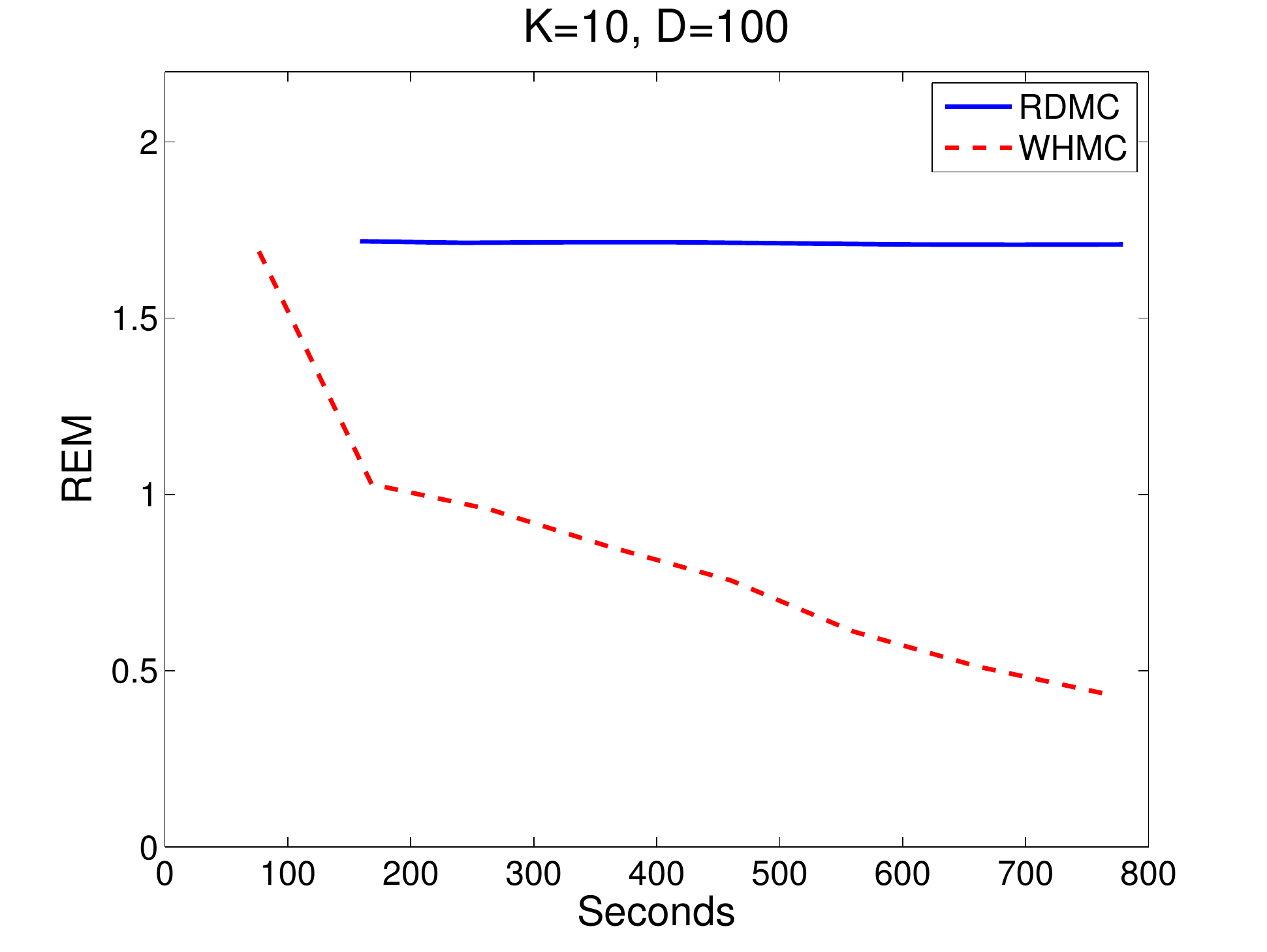}
}
\caption{Comparing WHMC to RDMC in terms of REM using $K=10$ mixtures of $D$-dimensional Gaussians with $D=20$ (left panel) and $D=100$ (right panel).}
\vspace{-18pt}
\label{unknownModes}
\end{center}
\end{figure}

\subsection{Mixture of Gaussians with Unknown Modes}

We now evaluate our method's performance in terms of searching for new modes and updating the network of wormholes.  For this example, we simulate a
mixture of 10 $D$-dimensional Gaussian distributions, with $D=10, 100$, and compare our method to RDMC. While RDMC runs four parallel HMC chains
initially to discover a subset of modes and to fit a truncated Gaussian distribution around each identified mode, we run four parallel optimizers
(different starting points) using the BFGS method. At regeneration times, each chain of RDMC uses the Dirichlet process mixture model to fit a new truncated
Gaussian around modes and possibly identify new modes. We on the other hand run the BGFS algorithm based on the residual energy function (with
$T=1.05$) to discover new modes for each chain. Figure \ref{unknownModes} shows WHMC reduces REM much faster than RDMC for both $D=10$ and $D=100$.
Here, the recorded time (horizontal axis) accounts for the computational overhead for adapting the transition kernels. For $D=10$, our method has a substantially lower REM compared to RDMC. For $D=100$, while our method identifies new modes over time and reduces REM substantially, RDMC fails to identify new modes so as a result its REM remains high over time. 

\section{Conclusions and Discussion}
We have proposed a new algorithm for sampling from multimodal distributions. Using empirical results, we have shown that our method performs well in high dimensions.

Our method involves several parameters that require tuning. However, these parameters can be adjusted at regeneration times without affecting the stationary distribution.

Although we used a flat base metric (i.e., $\boldsymbol I$) in the examples discussed in this paper, our method can be easily extended by specifying a more informative base metric (e.g., Fisher information) that adapts to local geometry.

\section*{Acknowledgements}
This material is based upon work supported by the National Science Foundation under Grant No. 1216045. 
We would like to thank Sungin Ahn for sharing his codes for the RDMC algorithm.

\newpage


\begin{thebibliography}{}

\bibitem[Ahn et~al.(2013)Ahn, Chen, and Welling]{ahn13}
S.~Ahn, Y.~Chen, and M.~Welling.
\newblock {Distributed and adaptive darting Monte Carlo through regenerations}.
\newblock In \emph{Proceedings of the 16th International Conference on
  Artificial Intelligence and Statistics (AI Stat)}, 2013.

\bibitem[Amari and Nagaoka(2000)]{amari00}
S.~Amari and H.~Nagaoka.
\newblock \emph{Methods of Information Geometry}, volume 191 of
  \emph{Translations of Mathematical monographs}.
\newblock Oxford University Press, 2000.

\bibitem[Braak(2006)]{braak06}
C.~J. F.~Ter Braak.
\newblock A markov chain monte carlo version of the genetic algorithm
  differential evolution: easy bayesian computing for real parameter spaces.
\newblock \emph{Statistics and Computing}, 16\penalty0 (3):\penalty0 239--249,
  2006.

\bibitem[Brockwell and Kadane(2005)]{brockwell05}
Anthony~E. Brockwell and Joseph~B. Kadane.
\newblock Identification of regeneration times in mcmc simulation, with
  application to adaptive schemes.
\newblock \emph{Journal of Computational and Graphical Statistics},
  14:\penalty0 436--458, 2005.

\bibitem[Celeux et~al.(2000)Celeux, Hurn, and Robert]{celeux00}
G.~Celeux, M.~Hurn, and C.~P. Robert.
\newblock Computational and inferential difficulties with mixture posterior
  distributions.
\newblock \emph{Journal of the American Statistical Association}, 95:\penalty0
  957--970, 2000.

\bibitem[Craiu et~al.(2009)Craiu, R., and Y.]{craiu09}
R.~V. Craiu, Jeffrey R., and Chao Y.
\newblock Learn from thy neighbor: Parallel-chain and regional adaptive mcmc.
\newblock \emph{Journal of the American Statistical Association}, 104\penalty0
  (488):\penalty0 1454--1466, 2009.

\bibitem[Duane et~al.(1987)Duane, Kennedy, Pendleton, and Roweth]{duane87}
S.~Duane, A.~D. Kennedy, B~J. Pendleton, and D.~Roweth.
\newblock {Hybrid Monte Carlo}.
\newblock \emph{Physics Letters B}, 195\penalty0 (2):\penalty0 216 -- 222,
  1987.

\bibitem[Gelfand and Dey(1994)]{gelfand94}
A.~E. Gelfand and D.~K. Dey.
\newblock Bayesian model choice: Asymptotic and exact calculation.
\newblock \emph{Journal of the Royal Statistical Society. Series B.},
  56\penalty0 (3):\penalty0 501--514, 1994.

\bibitem[Gilks et~al.(1998)Gilks, Roberts, and Sahu]{gilks98}
Walter~R. Gilks, Gareth~O. Roberts, and Sujit~K. Sahu.
\newblock Adaptive markov chain monte carlo through regeneration.
\newblock \emph{Journal of the American Statistical Association}, 93\penalty0
  (443):\penalty0 pp. 1045--1054, 1998.
\newblock ISSN 01621459.

\bibitem[Girolami and Calderhead(2011)]{girolami11}
M.~Girolami and B.~Calderhead.
\newblock {Riemann manifold Langevin and Hamiltonian Monte Carlo methods}.
\newblock \emph{Journal of the Royal Statistical Society, Series B}, (with
  discussion) 73\penalty0 (2):\penalty0 123--214, 2011.

\bibitem[Green(1995)]{green95}
Peter~J. Green.
\newblock Reversible jump markov chain monte carlo computation and bayesian
  model determination.
\newblock \emph{Biometrika}, 82:\penalty0 711--732, 1995.

\bibitem[Hinton et~al.(2004)Hinton, Welling, and Mnih]{hinton04}
G.~E. Hinton, M.~Welling, and A.~Mnih.
\newblock Wormholes improve contrastive divergence.
\newblock In \emph{Advances in Neural Information Processing Systems 16}, 2004.

\bibitem[Ihler et~al.(2005)Ihler, III, Moses, and Willsky]{Ihler05}
A.~T. Ihler, J.~W.~Fisher III, R.~L. Moses, and A.~S. Willsky.
\newblock Nonparametric belief propagation for self-localization of sensor
  networks.
\newblock \emph{IEEE Journal on Selected Areas in Communications}, 23\penalty0
  (4):\penalty0 809--819, 2005.

\bibitem[Kirkpatrick et~al.(1983)Kirkpatrick, Gelatt, and
  Vecchi]{kirkpatrick83}
S.~Kirkpatrick, C.~D. Gelatt, and M.~P. Vecchi.
\newblock {Optimization by Simulated Annealing}.
\newblock \emph{Science, Number 4598, 13 May 1983}, 220\penalty0
  (4598):\penalty0 671--680, 1983.

\bibitem[Kleinberg and Tardos(2005)]{kleinberg05}
J.~Kleinberg and E.~Tardos.
\newblock \emph{Algorithm Design}.
\newblock Addison-Wesley Longman Publishing Co., Inc., Boston, MA, USA, 2005.
\newblock ISBN 0321295358.

\bibitem[Lan et~al.(2012)Lan, Stathopoulos, Shahbaba, and Girolami]{lan12}
S.~Lan, V.~Stathopoulos, B.~Shahbaba, and M.~Girolami.
\newblock {Lagrangian Dynamical Monte Carlo}.
\newblock arxiv.org/abs/1211.3759, 2012.

\bibitem[Laskey and Myers(2003)]{laskey03}
K.~B. Laskey and J.~W. Myers.
\newblock {Population Markov Chain Monte Carlo}.
\newblock \emph{Machine Learning}, 50:\penalty0 175--196, 2003.

\bibitem[Leimkuhler and Reich(2004)]{leimkuhler04}
B.~Leimkuhler and S.~Reich.
\newblock \emph{Simulating Hamiltonian Dynamics}.
\newblock Cambridge University Press, 2004.

\bibitem[Liu(2001)]{liu01}
Jun~S. Liu.
\newblock \emph{Monte Carlo Strategies in Scientific Computing}, chapter
  Molecular Dynamics and Hybrid Monte Carlo.
\newblock Springer-Verlag, 2001.

\bibitem[Mykland et~al.(1995)Mykland, Tierney, and Yu]{mykland95}
Per Mykland, Luke Tierney, and Bin Yu.
\newblock Regeneration in markov chain samplers.
\newblock \emph{Journal of the American Statistical Association}, 90\penalty0
  (429):\penalty0 pp. 233--241, 1995.
\newblock ISSN 01621459.

\bibitem[Neal(1993)]{neal93}
R.~M. Neal.
\newblock \emph{Probabilistic Inference Using {Markov} Chain {Monte Carlo}
  Methods}.
\newblock Technical Report CRG-TR-93-1, Department of Computer Science,
  University of Toronto, 1993.

\bibitem[Neal(2001)]{neal01}
R.~M. Neal.
\newblock Annealed importance sampling.
\newblock \emph{Statistics and Computing}, 11\penalty0 (2):\penalty0 125--139,
  2001.

\bibitem[Neal(2010)]{neal10}
R.~M. Neal.
\newblock {MCMC using Hamiltonian dynamics}.
\newblock In S.~Brooks, A.~Gelman, G.~Jones, and X.~L. Meng, editors,
  \emph{Handbook of Markov Chain Monte Carlo}. Chapman and Hall/CRC, 2010.

\bibitem[Neal(1996)]{neal96a}
Radford~M. Neal.
\newblock \emph{Bayesian Learning for Neural Networks}.
\newblock Springer-Verlag New York, Inc., Secaucus, NJ, USA, 1996.
\newblock ISBN 0387947248.

\bibitem[Nummelin(1984)]{nummelin84}
Esa Nummelin.
\newblock \emph{General Irreducible Markov Chains and Non-Negative Operators},
  volume~83 of \emph{Cambridge Tracts in Mathematics}.
\newblock Cambridge University Press, 1984.

\bibitem[Rudoy and Wolfe(2006)]{rudoy06}
D.~Rudoy and P.~J. Wolfe.
\newblock Monte carlo methods for multi-modal distributions.
\newblock In \emph{Signals, Systems and Computers, 2006. ACSSC '06. Fortieth
  Asilomar Conference on}, pages 2019--2023, 2006.

\bibitem[Sminchisescu and Triggs(2002)]{sminchisescu02}
C.~Sminchisescu and B.~Triggs.
\newblock Building roadmaps of local minima of visual models.
\newblock In \emph{In European Conference on Computer Vision}, pages 566--582,
  2002.

\bibitem[Sminchisescu and Welling(2011)]{sminchisescu11}
C.~Sminchisescu and M.~Welling.
\newblock Generalized darting monte carlo.
\newblock \emph{Pattern Recognition}, 44\penalty0 (10-11), 2011.

\bibitem[Warnes(2001)]{warnes01}
G.~R. Warnes.
\newblock {The normal kernel coupler: An adaptive Markov Chain Monte Carlo
  method for efficiently sampling from multi-modal distributions}.
\newblock Technical Report Technical Report No. 395, University of Washington,
  2001.

\bibitem[Welling and Teh(2011)]{welling11}
M.~Welling and Y.~W. Teh.
\newblock {B}ayesian learning via stochastic gradient {L}angevin dynamics.
\newblock In \emph{Proceedings of the International Conference on Machine
  Learning}, 2011.

\end{thebibliography}

\newpage
\appendix
\begin{center}
{\Huge Appendix}
\end{center}

In this supplementary document, we provide details of our proposed Wormhole Hamiltonian Monte Carlo (WHMC) algorithm and prove its convergence to the stationary distribution. For simplicity, we assume that ${\bf G}(\vect\theta)\equiv {\bf I}$. Our results can be extended to more general Riemannian metrics. 

In what follows, we first prove the convergence of our method when an external vector field is added to the dynamics. Next, we prove the convergence for our final algorithm, where besides an external vector field, we include an auxiliary dimension along which the network of wormholes are constructed. Finally, we provide our algorithm to identify regeneration times. 

\section{Adjustment of Metropolis acceptance probability in WHMC with vector field}

As mentioned in the paper, in high dimensional problems with isolated modes, the effect of wormhole metric could diminish fast as the sampler leaves one mode towards another mode. To avoid this issue, we have extended our method by including a vector field, ${\bf f}({\vect\theta}, {\bf v})$, which depends on a vicinity function illustrated in Figure \ref{picvic}. The resulting dynamics facilitates movements between modes,
\begin{equation}\label{rmldxf}
\begin{aligned}
&\dot{\vect\theta} && = && {\bf v} + {\bf f}({\vect\theta}, {\bf v})\\
&\dot{\bf v} && = && - \nabla_{\vect\theta} U({\vect\theta})
\end{aligned}
\end{equation}

We solve \eqref{rmldxf} using the generalized leapfrog integrator \citep{leimkuhler04,girolami11}:
\begin{alignat}{3}
&{\bf v}^{(\ell+1/2)} && = {\bf v}^{(\ell)} - \frac{\eps}{2} \nabla_{\vect\theta} U({\vect\theta}^{(\ell)})  \label{xflmc:1}\\
&{\vect\theta}^{(\ell+1)} && = {\vect\theta}^{(\ell)} + \epsilon [{\bf v}^{(\ell+1/2)} + ({\bf f}({\vect\theta}^{(\ell)}, {\bf v}^{(\ell+1/2)})+{\bf f}({\vect\theta}^{(\ell+1)}, {\bf v}^{(\ell+1/2)}))/2]
\label{xflmc:0}\\
&{\bf v}^{(\ell+1)} && = {\bf v}^{(\ell+1/2)} - \frac{\eps}{2} \nabla_{\vect\theta} U({\vect\theta}^{(\ell+1)}) \label{xflmc:2}
\end{alignat}
where $\ell$ is the index for leapfrog steps, and ${\bf v}^{(\ell+1/2)}$ denotes the current value of ${\bf v}$ after half
a step of leapfrog. The implicit equation \eqref{xflmc:0} can be solved by the fixed point iteration.

The integrator \eqref{xflmc:1}-\eqref{xflmc:2} is time reversible and numerically stable; however, it is not volume preserving. To fix this issue, we can adjust the Metropolis acceptance probability with the Jacobian determinant of the mapping $\hat T$ given by
\eqref{xflmc:1}-\eqref{xflmc:2} in order to satisfy the detailed balance condition. Denote ${\bf z}=({\vect\theta},{\bf v})$. Given the corresponding Hamiltonian function, $H({\bf z})$, we define $\mathbb P(d{\bf z})=\exp(-H({\bf z}))d{\bf z}$ and prove the following proposition \citep[See ][]{green95}.
\begin{prop}[Detailed Balance Condition with determinant adjustment]\label{prop:dbwda}
Let ${\bf z}'=\hat T_{L}({\bf z})$ be the proposal according to some time reversible integrator $\hat T_{L}$ for dynamics \eqref{rmldxf}. Then the detailed balance condition holds given the following adjusted acceptance probability: 
\begin{equation}\label{adjacpt}
\tilde\alpha({\bf z},{\bf z'}) =\min\left\{1,\frac{\exp(-H({\bf z'}))}{\exp(-H({\bf z}))}|\det d\hat T_{L}|\right\}
\end{equation}
\end{prop}
\begin{proof}
To prove detailed balance, we need to show
\begin{equation}\label{dbvc}
\tilde\alpha({\bf z},{\bf z'}) \mathbb P(d{\bf z}) = \tilde\alpha({\bf z'},{\bf z}) \mathbb P(d{\bf z'})
\end{equation}
The following steps show that this condition holds:
\begin{equation*}
\begin{split}
\tilde\alpha({\bf z},{\bf z'}) \mathbb P(d{\bf z}) & =\min\left\{1,\frac{\exp(-H({\bf z'}))}{\exp(-H({\bf z}))}\left|\frac{d{\bf z'}}{d{\bf z}}\right|\right\} \exp(-H({\bf z}))d{\bf z}\\
& \overset{{\bf z}= \hat T_{L}^{-1}({\bf z'})}= \min\left\{\exp(-H({\bf z})),\exp(-H({\bf z'}))\left|\frac{d{\bf z'}}{d{\bf z}}\right|\right\} \left|\frac{d{\bf z}}{d{\bf z'}}\right|d{\bf z'}\\
& = \min\left\{1,\frac{\exp(-H({\bf z}))}{\exp(-H({\bf z'}))}\left|\frac{d{\bf z}}{d{\bf z'}}\right|\right\} \exp(-H({\bf z'})) d{\bf z'} = \tilde\alpha({\bf z'},{\bf z}) \mathbb P(d{\bf z'})
\end{split}
\end{equation*}
\end{proof}

\begin{figure}[t]
  \begin{center}
    \includegraphics[scale=0.5]{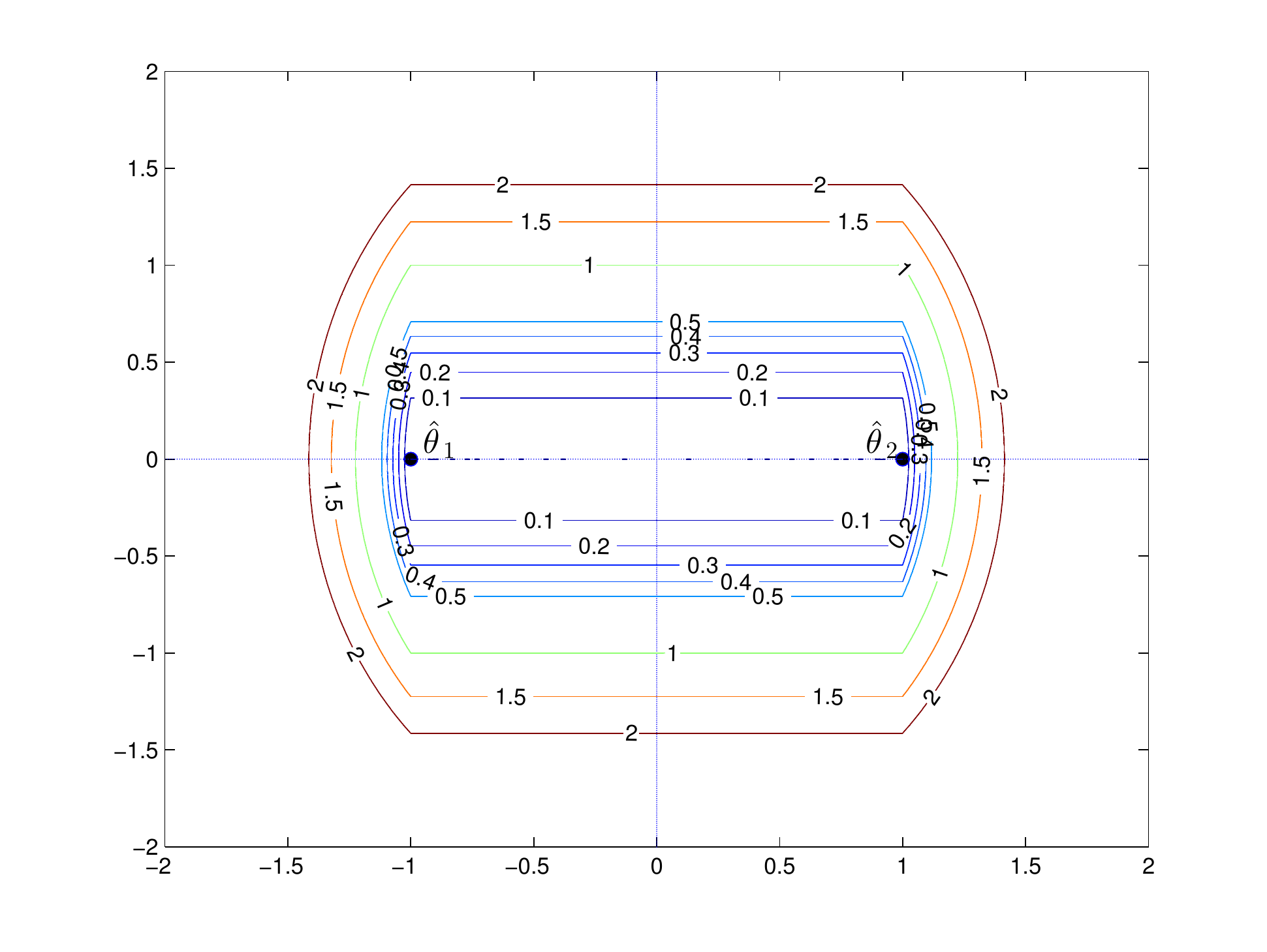}
  \end{center}
  \vspace{-10pt}
  \caption{Contour plot of the vicinity function $V(\vect\theta)$ in Equation (6) of our paper.}
  \label{picvic}
\end{figure}

We implement \eqref{xflmc:1}-\eqref{xflmc:2} for $L$ steps to generate a
proposal ${\bf z}^{(L+1)}$ and accept it with the following adjusted
probability:
\begin{equation*}
\alpha_{VF}({\bf z}^{(1)},{\bf z}^{(L+1)}) = \min\{1,\exp(-H({\bf z}^{(L+1)})+H({\bf z}^{(1)}))|\det {\bf J}_{VF}|\}
\end{equation*}
where the Jacobian determinant, $\det {\bf J}_{VF} = \prod_{n=1}^L \left|\dfrac{d{\bf z}^{(\ell+1)}}{d{\bf z}^{(\ell)}}\right|
=\prod_{n=1}^L \left|\dfrac{\pa({\vect\theta}^{(\ell+1)},{\bf v}^{(\ell+1)})}{\pa({\vect\theta}^{(\ell)},{\bf v}^{(\ell)})}\right|$, 
can be calculated through the following wedge product:
\begin{equation*}
d{\vect\theta}^{(\ell+1)} \wedge d{\bf v}^{(\ell+1)} = 
[{\bf I}-\frac{\eps}{2} \nabla_{{\vect\theta}^{\mathsf T}} {\bf f}({\vect\theta}^{(\ell+1)},{\bf v}^{(\ell+1/2)})]^{-1} [{\bf I}+\frac{\eps}{2} \nabla_{{\vect\theta}^{\mathsf T}} {\bf f}({\vect\theta}^{(\ell)},{\bf v}^{(\ell+1/2)})]\; d{\vect\theta}^{(\ell)} \wedge d{\bf v}^{(\ell)}
\end{equation*}
with $\nabla_{{\vect\theta}^{\mathsf T}} {\bf f}({\vect\theta},{\bf v}) = {\bf
v}^*_W({\bf v}^*_W)^{\mathsf T}{\bf v} \nabla\mathfrak m({\vect\theta})^{\mathsf T}$ \citep[See][for more details.]{lan12}

\section{WHMC in the augmented $D+1$ dimensional space}
Suppose that the current position, $\tilde{\vect\theta}$, of the sampler is near a mode denoted as $\tilde{\vect\theta}^{*}_{0}$. A
network of wormholes connects this mode to all the modes in the opposite world $\tilde{\vect\theta}^{*}_{k}, k=1,\cdots K$. Wormholes in the augmented space starting from this mode may still interfere each other since they intersect. To resolve this issue, instead of
deterministically weighing wormholes by the vicinity function (6), we use the following random
vector field $\tilde{\bf f}(\tilde{\vect\theta}, \tilde{\bf v})$:
\begin{equation*}
\tilde{\bf f}(\tilde{\vect\theta}, \tilde{\bf v})\sim
\begin{dcases}
(1-\sum_k \mathfrak m_k(\tilde{\vect\theta})) \delta_{\tilde{\bf v}}(\cdot) +
\sum_k \mathfrak m_k(\tilde{\vect\theta}) \delta_{2(\tilde{\vect\theta}^{*}_{k}-\tilde{\vect\theta})/e}(\cdot), & \;
\textrm{if}\, \sum_k \mathfrak m_k(\tilde{\vect\theta})<1 \\
\frac{\sum_k \mathfrak m_k(\tilde{\vect\theta}) \delta_{2(\tilde{\vect\theta}^{*}_{k}-\tilde{\vect\theta})/e}(\cdot)}{\sum_k
\mathfrak m_k(\tilde{\vect\theta})}, & \; \textrm{if}\, \sum_k \mathfrak m_k(\tilde{\vect\theta})\geq 1
\end{dcases}
\end{equation*}
where $e$ is the stepsize, $\delta$ is the Kronecker delta function, and $\mathfrak m_k(\tilde{\vect\theta}) =
\exp\{-V_k(\tilde{\vect\theta})/(DF)\}$. Here, the vicinity function $V_k(\tilde{\vect\theta})$ along the $k$-th wormhole is defined similarly to Equation (6),
\begin{equation*}
V_k(\tilde{\vect\theta})= \langle \tilde{\vect\theta}-\tilde{\vect\theta}^{*}_{0}, \tilde{\vect\theta}-\tilde{\vect\theta}^{*}_{k}\rangle + |\langle \tilde{\vect\theta}-\tilde{\vect\theta}^{*}_{0},\tilde{\bf v}^*_{W_{k}}\rangle||\langle \tilde{\vect\theta}-\tilde{\vect\theta}^{*}_{k},\tilde{\bf v}^*_{W_{k}}\rangle|
\end{equation*}
where $\tilde{\bf v}^*_{W_{k}} = (\tilde{\vect\theta}^{*}_{k}-\tilde{\vect\theta}^{*}_{0})/\Vert \tilde{\vect\theta}^{*}_{k}-\tilde{\vect\theta}^{*}_{0} \Vert$.


For each update, $\tilde{\bf f}(\tilde{\vect\theta}, \tilde{\bf v})$ is either set to $\tilde{\bf v}$ or
$2(\tilde{\vect\theta}^{*}_{k}-\tilde{\vect\theta})/e$ according to the position dependent probabilities defined in terms of $\mathfrak m_k(\tilde{\vect\theta})$. 
Therefore, we write the Hamiltonian dynamics in the extended space as follows:
\begin{equation}\label{rmldrandxf}
\begin{aligned}
&\dot{\tilde{\vect\theta}} && = && \tilde{\bf f}(\tilde{\vect\theta}, \tilde{\bf v})\\
&\dot{\tilde{\bf v}} && = && - \nabla_{\tilde{\vect\theta}} U(\tilde{\vect\theta})\end{aligned}
\end{equation}
We use \eqref{xflmc:1}-\eqref{xflmc:2} to numerically solve \eqref{rmldrandxf},
but replace \eqref{xflmc:0} with the following equation: 
\begin{equation}\label{stochevol}
\tilde{\vect\theta}^{(\ell+1)} = \tilde{\vect\theta}^{(\ell)} + e/2 [ \tilde{\bf f}(\tilde{\vect\theta}^{(\ell+1)}, \tilde{\bf v}^{(\ell+1/2)})+\tilde{\bf f}(\tilde{\vect\theta}^{(\ell)}, \tilde{\bf v}^{(\ell+1/2)})]
\end{equation}
Note that this is an implicit equation, which can be solved using the fixed point iteration approach \citep{leimkuhler04, girolami11}.

According to the above dynamic, at each leapfrog step, $\ell$, the sampler either stays at the vicinity of $\tilde{\vect\theta}^{*}_{0}$ or
proposes a move towards a mode $\tilde{\vect\theta}^{*}_{k}$ in the opposite world depending on the values of $\vect{\tilde
f}(\tilde{\vect\theta}^{(\ell)} \tilde{\bf v}^{(\ell+1/2)})$ and $\tilde{\bf f}(\tilde{\vect\theta}^{(\ell+1)},
\tilde{\bf v}^{(\ell+1/2)})$. For example, if $\tilde{\bf f}(\tilde{\vect\theta}^{(\ell)}, \tilde{\bf v}^{(\ell+1/2)}) =
2(\tilde{\vect\theta}^{*}_{k}-\tilde{\vect\theta}^{(\ell)})/e$, and $\tilde{\bf f}(\tilde{\vect\theta}^{(\ell+1)},
\tilde{\bf v}^{(\ell+1/2)})=\tilde{\bf v}^{(\ell+1/2)}$, then equation \eqref{stochevol} becomes
\begin{equation*}
\tilde{\vect\theta}^{(\ell+1)} = \tilde{\vect\theta}^{*}_{k} + \frac{e}{2} \tilde{\bf v}^{(\ell+1/2)}
\end{equation*}
which indicates that a move to the $k$-th mode in the opposite wold has in fact occurred. Note that the movement
$\tilde{\vect\theta}^{(\ell)}\to \tilde{\vect\theta}^{(\ell+1)}$ in this case is discontinuous since $$\lim_{e\to 0}\Vert
\tilde{\vect\theta}^{(\ell+1)} - \tilde{\vect\theta}^{(\ell)}\Vert \geq 2h >0$$
Therefore, in such cases, there will be an energy gap, $\Delta E := H(\tilde{\vect\theta}^{(\ell+1)}, \tilde{\bf v}^{(\ell+1)}) -
H(\tilde{\vect\theta}^{(\ell)}, \tilde{\bf v}^{(\ell)})$, between the two states. We need to adjust the Metropolis acceptance probability
to account for the resulting energy gap.\footnote{Proposition \ref{prop:dbwda} does not apply because Jacobian determinant is not well defined for discontinuous movement.} Further, we limit the maximum number of jumps within each iteration of MCMC (i.e., over $L$ leapfrog steps) to 1 so the sampler can explore the vicinity of the new mode before making another jump. Algorithm \ref{Alg:WHMC} provides the details of this approach.

We note that according to the definition of $\tilde{\bf f}(\tilde{\vect\theta}, \tilde{\bf v})$ and equation \eqref{stochevol}, the
jump occurs randomly. We use $\ell'$ to denote the step at which the sampler jumps. That is, $\ell'$ randomly takes a value in $\{0,1,\cdots,L\}$ depending on $\tilde{\bf f}$.
When there is no jump along the trajectory, we set $\Delta E=0$, $\ell'=0$, and the algorithm reduces to standard HMC.
In the following, we first prove the detailed balance condition when a jump happens, and then use it to prove the convergence of the algorithm to the stationary distribution.

When a mode jumping occurs at
some fixed step $\ell'$, we can divide the $L$ leapfrog steps into three parts: $\ell'-1$
steps continuous movement according to standard HMC, 1 step discontinuous jump, 
and $L-\ell'$ steps according to standard HMC in the opposite world. Note that the Metropolis acceptance probability in standard HMC can be written as
\begin{equation*}
\alpha(\tilde{\bf z}^{(L+1)},\tilde{\bf z}^{(1)}) = \min\{1,\exp(-H(\tilde{\bf
z}^{(L+1)})+H(\tilde{\bf z}^{(1)}))\} =
\min\left\{1,\exp\left(-\sum_{\ell=1}^{L} (H(\tilde{\bf
z}^{(\ell+1)})-H(\tilde{\bf z}^{(\ell)}))\right)\right\}
\end{equation*}
Each summand $H(\tilde{\bf z}^{(\ell+1)})-H(\tilde{\bf z}^{(\ell)})$ is small ($\mathcal
O(\eps^3)$) except for the $\ell'$-th one, where there is an energy gap, $\Delta E$.
Given that the jump happens at the $\ell'$-th step, we should remove
the $\ell'$-th summand from the acceptance probability:
\begin{align}
\alpha_{RVF}(\tilde{\bf z}^{(1)},\tilde{\bf z}^{(L+1)}) &=
\min\left\{1,\exp\left(-\sum_{\ell\neq \ell'} (H(\tilde{\bf
z}^{(\ell+1)})-H(\tilde{\bf z}^{(\ell)}))\right)\right\} \nonumber\\
&= \min\left\{1,\frac{\exp(-H(\tilde{\bf z}^{(L+1)}))}{\exp(-H(\tilde{\bf
z}^{(\ell'+1)}))}\frac{\exp(-H(\tilde{\bf z}^{(\ell')}))}{\exp(-H(\tilde{\bf
z}^{(1)}))}\right\} \label{apHMCcont}\\
&=\min\{1,\exp(-H(\tilde{\bf z}^{(L+1)})+H(\tilde{\bf z}^{(1)})+\Delta E)\} \nonumber
\end{align}
That is, we only count the acceptance probability for the
first $\ell'-1$ and the last $L-\ell'$ steps of continuous movement in standard
HMC; at $\ell'+1$ step, we ``reset" the energy level by accounting for the energy
gap $\Delta E$. Therefore, the following proposition is true.
\begin{prop}[Detailed Balance Condition with energy adjustment]\label{prop:dbwea}
When a discontinuous jump happens in WHMC, we have the following detailed balance condition with the adjusted acceptance probability \eqref{apHMCcont}.
\begin{equation}\label{jtDB}
\alpha_{RVF}(\tilde{\bf z}^{(1)},\tilde{\bf z}^{(L+1)}) \mathbb P(d\tilde{\bf
z}^{(1)}) \mathbb P(d\tilde{\bf z}^{(\ell'+1)}) = \alpha_{RVF}(\tilde{\bf
z}^{(L+1)},\tilde{\bf z}^{(1)}) \mathbb P(d\tilde{\bf z}^{(L+1)}) \mathbb
P(d\tilde{\bf z}^{(\ell')})
\end{equation}
\end{prop}

Next, we use $\hat T_{\ell}: \tilde{\bf z}^{(\ell)}\mapsto \tilde{\bf z}^{(\ell+1)}$
to denote an integrator consisting of \eqref{xflmc:1}\eqref{stochevol}\eqref{xflmc:2}. Note that
$\hat T_{\ell}$ depends on $\tilde{\bf f}$.
Denote $\hat T_{1:\ell}:=\hat T_{\ell}\circ\cdots\circ\hat T_{1}$.
Following \cite{liu01}, we can prove the following stationarity theorem.
\begin{thm}[Stationarity of WHMC]
The samples given by WHMC (algorithm \ref{Alg:WHMC}) have the target distribution $\pi(\cdot)$ as its stationarity distribution.
\end{thm}
\begin{proof}
Let $\tilde{\bf z}^*=\hat T_{1:L}(\tilde{\bf z})$.
Suppose $\tilde{\vect\theta}^*\sim f(\tilde{\vect\theta}^*)$.
We want to prove that $f(\cdot) = \pi(\cdot)$ through
$\E_f[h(\tilde{\vect\theta}^*)]=\E_{\pi}[h(\tilde{\vect\theta}^*)]$ for any square
integrable function $h(\tilde{\vect\theta}^*)$.

Note that $\tilde{\vect\theta}^*$ is either an
accepted proposal or the current state after a rejection. Therefore,
\begin{equation*}
\begin{split}
&\E_f[h(\tilde{\vect\theta}^*)] \\
= & \iiint h(\tilde{\vect\theta}^*) [\alpha(\hat T_{1:L}^{-1}(\tilde{\bf
z}^*), \tilde{\bf z}^*) \mathbb P(d\hat T_{1:L}^{-1}(\tilde{\bf z}^*)) \mathbb P(d\hat
T_{1+\ell':L}^{-1}(\tilde{\bf z}^*)) + (1-\alpha(\tilde{\bf z}^*,\hat
T_{1:L}(\tilde{\bf z}^*))) \mathbb P(d\tilde{\bf z}^*) \mathbb P(d\hat
T_{1:\ell'}(\tilde{\bf z}^*))] d\ell'\\
= & \int h(\tilde{\vect\theta}^*) \mathbb P(d\tilde{\bf z}^*) +\\
& \iiint h(\tilde{\vect\theta}^*)[\alpha(\hat T_{1:L}^{-1}(\tilde{\bf z}^*),
\tilde{\bf z}^*)\mathbb P(d \hat T_{1:L}^{-1}(\tilde{\bf z}^*)) \mathbb P(d\hat
T_{1+\ell':L}^{-1}(\tilde{\bf z}^*)) -
\alpha(\tilde{\bf z}^*,\hat T_{1:L}(\tilde{\bf z}^*)) \mathbb P(d\tilde{\bf
z}^*) \mathbb P(d\hat T_{1:\ell'}(\tilde{\bf z}^*))] d\ell'
\end{split}
\end{equation*}
Therefore, it suffices to prove that
\begin{equation}\label{boilsdown}
\iiint h(\tilde{\vect\theta}^*) \alpha(\hat T_{1:L}^{-1}(\tilde{\bf z}^*), \tilde{\bf z}^*) \mathbb P(d\hat T_{1:L}^{-1}(\tilde{\bf z}^*)) \mathbb P(d\hat T_{1+\ell':L}^{-1}(\tilde{\bf z}^*)) d\ell'
= \iiint h(\tilde{\vect\theta}^*) \alpha(\tilde{\bf z}^*,\hat T_{1:L}(\tilde{\bf z}^*)) \mathbb P(d\tilde{\bf z}^*) \mathbb P(d\hat T_{1:\ell'}(\tilde{\bf z}^*)) d\ell'
\end{equation}

Based on its construction, $\hat T_{\ell}$ is time reversible for all $\ell$.
Denote the involution $\nu: (\tilde{\vect\theta},\tilde{\bf v})\mapsto
(\tilde{\vect\theta},-\tilde{\bf v})$. We have $\hat T^{-1}_{\ell}(\tilde{\bf
z}^*)=\nu\hat T_{\ell}\nu(\tilde{\bf z}^*)$.
Further, because $E$ is quadratic in $\tilde{\bf v}$,  we have $H(\nu(\tilde{\bf z}))=H(\tilde{\bf
z})$. Therefore $\alpha(\nu(\tilde{\bf z}),\tilde{\bf z}')=\alpha(\tilde{\bf z},\nu(\tilde{\bf z}'))$.
Then the left hand side of \eqref{boilsdown} becomes
\begin{equation*}
\begin{split}
\textrm{LHS} & = \iiint h(\tilde{\vect\theta}^*) \alpha(\nu\hat
T_{1:L}\nu(\tilde{\bf z}^*), \tilde{\bf z}^*) \mathbb P(d\nu\hat
T_{1:L}\nu(\tilde{\bf z}^*)) \mathbb P(d\nu\hat T_{1+\ell':L}\nu(\tilde{\bf
z}^*)) d\ell'\\
& = \iiint h(\tilde{\vect\theta}^*) \alpha(\hat
T_{1:L}\nu(\tilde{\bf z}^*), \nu(\tilde{\bf z}^*)) \mathbb P(d\hat
T_{1:L}\nu(\tilde{\bf z}^*)) \mathbb P(d\hat T_{1+\ell':L}\nu(\tilde{\bf
z}^*)) d\ell'\\
& \overset{\nu(\tilde{\bf z}^*)\mapsto\tilde{\bf z}^*}= \iiint
h(\tilde{\vect\theta}^*) \alpha(\hat T_{1:L}(\tilde{\bf z}^*), \tilde{\bf z}^*) \mathbb P(d\hat
T_{1:L}(\tilde{\bf z}^*)) \mathbb P(d\hat T_{1+\ell':L}(\tilde{\bf
z}^*)) d\ell'
\end{split}
\end{equation*}
On the other hand, by the detailed balance condition \eqref{jtDB}, the right hand
side of \eqref{boilsdown} becomes
\begin{equation*}
\textrm{RHS}
 = \iiint h(\tilde{\vect\theta}^*) \alpha(\hat
T_{1:L}(\tilde{\bf z}^*), \tilde{\bf z}^*) \mathbb P(d\hat
T_{1:L}(\tilde{\bf z}^*)) \mathbb P(d\hat T_{1:\ell'-1}(\tilde{\bf
z}^*)) d\ell'
\end{equation*}
Note that $\mathbb P(d\hat T_{1+\ell':L}(\tilde{\bf z}^*)) = \mathbb P(d\hat T_{1:\ell'-1}(\tilde{\bf
z}^*))$ since both $\hat T_{1+\ell':L}$ and $\hat T_{1:\ell'-1}$ are leapfrog
steps of standard HMC (no jump). The difference in the numbers of leapfrog steps
(i.e., $L-\ell'$ and $\ell'-1$) does not affect the stationarity since the number of leapfrog steps can be randomized in HMC \citep{neal10}. Therefore, we have $\textrm{LHS}=\textrm{RHS}$, which proves \eqref{boilsdown}.
\end{proof}


\begin{algorithm}[htpb]
\caption{Wormhole Hamilton Monte Carlo (WHMC)}
\label{Alg:WHMC}
\begin{algorithmic}
\STATE Prepare the modes ${\vect\theta}^{*}_{k}, \;k=1,\cdots K$
\STATE Set $\tilde{\vect\theta}^{(1)} = \textrm{current}\; \tilde{\vect\theta}$
\STATE Sample velocity $\tilde{\bf v}^{(1)}\sim \mathcal N(0,{\bf I}_{D+1})$
\STATE Calculate ${\bf E}(\tilde{\vect\theta}^{(1)},\tilde{\bf v}^{(1)}) = U(\tilde{\vect\theta}^{(1)}) + K(\tilde{\bf v}^{(1)})$
\STATE Set $\Delta \log\det = 0$, $\Delta E= 0 $, $\mathrm{Jumped} = \mathrm{false}$. 
\FOR{$n=1$ to $L$}
\STATE $\tilde{\bf v}^{(\ell+\frac{1}{2})} = \tilde{\bf v}^{(\ell)} - \frac{e}{2} \nabla_{\tilde{\vect\theta}} U(\tilde{\vect\theta}^{(\ell)})$
\IF{$\mathrm{Jumped}$}
\STATE $\tilde{\vect\theta}^{(\ell+1)} = \tilde{\vect\theta}^{(\ell)} + e \tilde{\bf v}^{(\ell+\frac{1}{2})}$
\ELSE
\STATE Find the closest mode $\tilde{\vect\theta}^{*}_0$ and build a
network connecting it to all modes $\tilde{\vect\theta}^{*}_{k}, \;k=1,\cdots K$ in the opposite world
\FOR{$m=1$ to $M$}
\STATE Calculate $\mathfrak m_k(\hat{\tilde{\vect\theta}}^{(m)}), k=1,\cdots K$
\STATE Sample $u\sim \mathrm{Unif}(0,1)$
\IF{$u<1-\sum_k \mathfrak m_k(\hat{\tilde{\vect\theta}}^{(m)})$}
\STATE Set ${\bf f}(\hat{\tilde{\vect\theta}}^{(m)}, \tilde{\bf v}^{(\ell+\frac{1}{2})}) = \tilde{\bf v}^{(\ell+\frac{1}{2})}$
\ELSE
\STATE Choose one of the $k$ wormholes according to probability $\{\mathfrak m_k/\sum_{k'} \mathfrak m_{k'}\}$
 and set ${\bf f}(\hat{\tilde{\vect\theta}}^{(m)}, \tilde{\bf
v}^{(\ell+\frac{1}{2})}) = 2(\tilde{\vect\theta}^{*}_{k}-\hat{\tilde{\vect\theta}}^{(m)})/e$
\ENDIF
\STATE $\hat{\tilde{\vect\theta}}^{(m+1)} = \tilde{\vect\theta}^{(\ell)} + \frac{e}{2}[{\bf f}(\hat{\tilde{\vect\theta}}^{(m)}, \tilde{\bf v}^{(\ell+\frac{1}{2})}) + {\bf f}(\tilde{\vect\theta}^{(\ell)}, \tilde{\bf v}^{(\ell+\frac{1}{2})})]$
\ENDFOR
\STATE $\tilde{\vect\theta}^{(\ell+1)} = \hat{\tilde{\vect\theta}}^{(M+1)}$
\ENDIF
\STATE $\tilde{\bf v}^{(\ell+1)} = \tilde{\bf v}^{(\ell+\frac{1}{2})} - \frac{e}{2} \nabla_{\tilde{\vect\theta}} U(\tilde{\vect\theta}^{(\ell+1)})$
\STATE If a jump has occurred, set $\mathrm{Jumped}=\mathrm{true}$ and calculate energy gap $\Delta E$.
\ENDFOR
\STATE Calculate ${\bf E}(\tilde{\vect\theta}^{(L+1)},\tilde{\bf v}^{(L+1)}) = U(\tilde{\vect\theta}^{(L+1)}) + K(\tilde{\bf v}^{(L+1)})$
\STATE $p = \exp\{-{\bf E}(\tilde{\vect\theta}^{(L+1)},\tilde{\bf v}^{(L+1)})+{\bf E}(\tilde{\vect\theta}^{(1)},\tilde{\bf v}^{(1)})  + \Delta E\}$
\STATE Accept or reject the proposal $(\tilde{\vect\theta}^{(L+1)},\tilde{\bf v}^{(L+1)})$ according to $p$
\end{algorithmic}
\end{algorithm}

\begin{algorithm}[h]
\caption{Regeneration in Wormhole Hamiltonian Monte Carlo}
\label{Alg:RWHMC}
\begin{algorithmic}
\STATE Initially search modes $\hat{\vect\theta}_{1},\cdots, \hat{\vect\theta}_{k}$
\FOR{$n=1$ to $L$}
\STATE Sample $\tilde{\vect\theta}=(\vect\theta,\theta_{D+1})$ as the current state according to WHMC (algorithm \ref{Alg:WHMC}).
\STATE Fit a mixture of Gaussians $q(\vect\theta)$ with known modes, Hessians and relative weights. Propose ${\vect\theta}^{*}\sim q(\cdot)$ and accept it with probability $\alpha =
\min\left\{1,\frac{\pi({\vect\theta}^{*})/q({\vect\theta}^{*})}{\pi(\vect\theta)/q(\vect\theta)}\right\}$.
\IF{${\vect\theta}^{*}$ accepted}
\STATE Determine if ${\vect\theta}^{*}$ is a regeneration using (9)-(12) with
${\vect\theta}_{t}=\vect\theta$ and ${\vect\theta}_{t+1}={\vect\theta}^{*}$.
\IF{Regeneration occurs}
\STATE Search new modes by minimizing $U_{\bf r}(\vect\theta,T)$; if new modes are discovered, update the mode library, wormhole network, and $q(\vect\theta)$.\\
\STATE Discard ${\vect\theta}^{*}$, sample ${\vect\theta}^{(\ell+1)}\sim Q(\cdot)$ as in (12) using rejection sampling.
\ELSE
\STATE Set ${\vect\theta}^{(\ell+1)} = {\vect\theta}^{*}$.
\ENDIF
\ELSE
\STATE Set ${\vect\theta}^{(\ell+1)} = \tilde{\vect\theta}$.
\ENDIF
\ENDFOR
\end{algorithmic}
\end{algorithm}

\end{document}